\documentclass[format=acmsmall, review=false, screen]{acmart}

\makeatletter
\def\@ACM@checkaffil{
    \if@ACM@instpresent\else
    \ClassWarningNoLine{\@classname}{No institution present for an affiliation}%
    \fi
    \if@ACM@citypresent\else
    \ClassWarningNoLine{\@classname}{No city present for an affiliation}%
    \fi
    \if@ACM@countrypresent\else
        \ClassWarningNoLine{\@classname}{No country present for an affiliation}%
    \fi
}
\makeatother

\acmConference{}{}{} 
\startPage{0} 
\AtBeginDocument{
\fancypagestyle{firstpagestyle}{
	\fancyhf{}%
	\fancyfoot[L]{\footnotesize{}}%
}
\fancypagestyle{ourstyle}{
	\fancyhf{}%
	\fancyhead[C]{\footnotesize{\mbox{\shortauthors}}}
	\fancyhead[R]{\thepage}%
}
\pagestyle{ourstyle}
\renewcommand{\footnotetextcopyrightpermission}[1]{} 
\renewcommand{\footnotetextauthorsaddresses}[1]{} 
}

\usepackage{booktabs} 
\usepackage[ruled]{algorithm2e} 

\SetAlFnt{\small}
\SetAlCapFnt{\small}
\SetAlCapNameFnt{\small}
\SetAlCapHSkip{0pt}
\IncMargin{-\parindent}
\allowdisplaybreaks

\setcitestyle{authoryear}

\title[Avoiding Overrepresentation in Multi-Winner Voting]{Beyond Lower Quota: \\ Avoiding Overrepresentation in Multi-Winner Voting}

\author{Anton Baychkov}
\affiliation{%
  \institution{University of Warwick}
  \city{}
  \country{United Kingdom}}

\author{Martin Lackner}
\affiliation{%
  \institution{University of Applied Sciences St. Pölten}
  \city{}
  \country{Austria}}

\author{Jan Maly}
\affiliation{%
  \institution{WU Vienna University of Economics and Business}
  \city{}
  \country{Austria}}

\author{Oliviero Nardi}
\affiliation{%
	\institution{TU Wien \& WU Vienna University of Economics and Business}
  \city{}
  \country{Austria}}

\author{Jannik Peters}
\affiliation{%
  \institution{Shanghai University of Finance and Economics}
  \city{}
  \country{Shanghai}}

\begin{abstract}
	In the recent social choice literature on proportional representation, much attention has been given to the question of avoiding \emph{underrepresentation} in approval-based multi-winner voting. In this paper, we explore the largely overlooked complementary question of avoiding \emph{overrepresentation}. This has not been explored systematically, despite being a desirable property with concrete applications. Intuitively, overrepresentation happens when a group determines a disproportionately large part of the committee, thereby exceeding the group's \emph{quota}.

	We formulate a strong and appealing axiom for avoiding overrepresentation, called \emph{justified upper quota} (JUQ). We introduce a generalization of Thiele rules, \emph{composite Thiele rules}, and characterize the unique rule in this class satisfying our axiom. This rule, \emph{Adams-AV}, which naturally extends Adams' apportionment method, has not been studied before. Additionally, we introduce a polynomial-time rule that satisfies JUQ.

	Furthermore, we introduce \emph{justified near quota}, an axiom that balances avoiding under- and overrepresentation. It characterizes the unique Thiele rule extending the Sainte-Laguë apportionment method. Finally, we analyze the compatibility of our axioms with established proportionality notions such as EJR+.
\end{abstract}

\usepackage{nicefrac}
\usepackage{tikz}
\usepackage{tikz-cd}
\usepackage{cleveref}
\usepackage{thmtools, thm-restate} 
\usepackage{enumitem}
\usepackage{wrapfig}
\usepackage{wrapstuff}
\usepackage{arydshln} 

\theoremstyle{acmplain}
\newtheorem{theorem}{Theorem}
\newtheorem{lemma}[theorem]{Lemma}
\newtheorem{corollary}[theorem]{Corollary}
\newtheorem{proposition}[theorem]{Proposition}
\newtheorem{observation}[theorem]{Observation}

\declaretheorem[
  style=acmdefinition,
  title=Example,
  refname={example,examples},
  Refname={Example,Examples},
  qed={$\vartriangle$}
]{example}

\declaretheorem[
  style=acmdefinition,
  title=Definition,
  refname={definition,definitions},
  Refname={Definition,Definitions},
  qed={$\vartriangle$}
]{definition}

\newcommand{\nats}{\mathbb{N}}
\newcommand{\posNats}{\nats_{>0}}
\newcommand{\rationals}{\mathbb{Q}}
\newcommand{\reals}{\mathbb{R}}
\newcommand{\nonNegReals}{\reals_{\geq0}}
\newcommand{\posReals}{\reals_{>0}}

\newcommand{\indicator}{\mathbf{1}}

\newcommand{\NP}{\textsc{NP}}

\newcommand{\powerset}[1]{\mathcal{P}(#1)}
\newcommand{\npowerset}[2]{\mathcal{P}_{#1}(#2)}

\DeclareMathOperator*{\argmax}{arg\,max}

\newcommand{\voters}{N}
\newcommand{\candidates}{C}
\newcommand{\app}{A}
\newcommand{\approfile}{\boldsymbol{\app}}
\newcommand{\kcommittees}{\npowerset{k}{\candidates}}
\newcommand{\score}{\mathrm{score}}
\newcommand{\ceil}[1]{\lceil #1 \rceil}
\newcommand{\floor}[1]{\lfloor #1 \rfloor}

\usepackage{xcolor}
\usepackage{pifont}
\newcommand{\cmark}{\textcolor{green!70!black}{\ding{51}}}  

\begin{document}

\begin{titlepage}

\maketitle

\vspace{1cm}
\setcounter{tocdepth}{1} 
\tableofcontents

\end{titlepage}

\section{Introduction}

Consider a scenario where a group of agents must jointly select a certain number of items from a pool of alternatives. Many real-world situations can be described in this way; for example, committee elections~\citep{LaSk22a}, parliamentary apportionment~\citep{BaYo01a}, or finding group recommendations~\citep{SFL16a}.

This problem, known as \emph{multi-winner voting}~\citep{FSST17a,LaSk22a}, has become one of the most-studied settings in (computational) social choice~\citep{BCE+14a}. Most of the effort has focused on \emph{approval-based multi-winner voting}, where voters express their preferences by \emph{approving} candidates \citep{LaSk22a}. In particular, one of the main issues in this area is how to guarantee \emph{proportional representation}, which has so far largely been equated with avoiding underrepresentation. Briefly, the idea is that if a group of voters is ``large enough'', then it deserves to be represented in the outcome. More precisely, if a certain fraction of the voting population has similar preferences, they should be represented by (i.e., approve of) at least roughly the same fraction of candidates in the winning committee. Much progress has been made in this direction, and many properties (\emph{axioms}) have been proposed that formally capture this notion (e.g., extended justified representation~\citep{ABC+16a}, priceability~\citep{PeSk20b}, and many others \citep{BrPe23a,PPS21a,KKL25a}). Furthermore, there are concrete voting rules that embody this property, for instance, Proportional Approval Voting (PAV)~\citep{LaSk21a,Thie95a}, the Method of Equal Shares~\citep{PeSk20b}, or Phragmén's sequential rule~\citep{BFJL24a}.

\begin{wrapfigure}{r}{0.35\textwidth}
  \centering
  \includegraphics[width=0.33\textwidth]{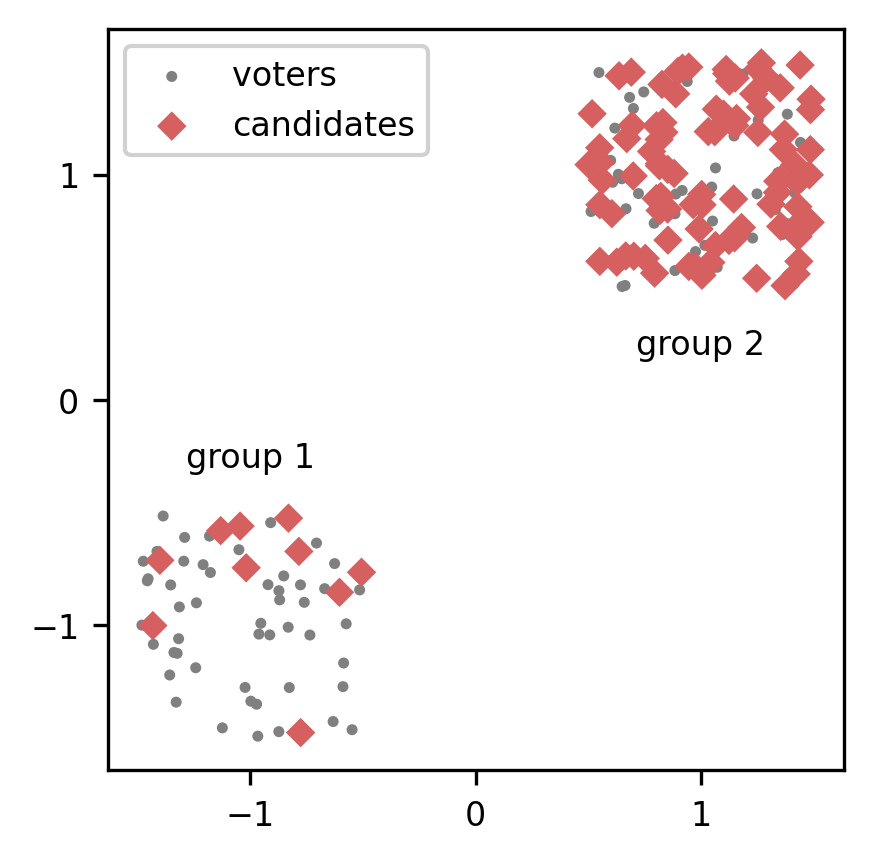}
  \caption{Two groups with the same number of voters.}\label{fig:unbalanced}
\end{wrapfigure}
The main goal of our paper is to look at the other side of the coin, which has not received as much attention: \emph{overrepresentation}. In short, this happens when a group of voters receives a fraction of the winning committee that is much larger than what its size would warrant. 
To illustrate this, consider the example shown in Figure~\ref{fig:unbalanced}. Here, we assume that voters and candidates correspond to two-dimensional points and voters prefer candidates close to them, a type of spatial model that is standard in voting theory. While the groups of voters are of equal size (50 each), there are 10 candidates associated with group~1 and 90 with group~2 (all points were sampled uniformly in either $[0.5, 1.5]^2$ or $[-1.5, -0.5]^2$). We are interested in selecting $10$ candidates, and voters approve the $10$ candidates closest to them. Consequently, the approval preferences of voters in group 1 are identical, whereas those in group~2 are diverse.\footnote{A similar scenario can occur even when candidates are distributed evenly across the two groups, if one group---potentially for strategic reasons---coordinates to vote for the same $10$ candidates. Everyone in this group submits the same ballot (a ``bullet vote''), thus achieving high cohesiveness in the sense of justified representation, as defined by \citet{ABC+16a}.} How do established, \emph{proportional} voting rules handle this scenario? When averaging over 1000 sampled instances, group~2 receives on average
$3.5$ candidates with PAV, between $1.7$ and $3.4$ with the Method of Equal Shares (depending on the completion method), and $3.4$ with Phragmén's sequential rule.
We see that these proportional rules lead to \emph{overrepresentation of group~1}.\footnote{One might wonder why this does not constitute an example of \emph{under}representation for group~2. Indeed, the preferences of group~2 are too diverse to be regarded as a well-defined group with cohesive preferences; instead, there are many subgroups that are too small to be considered by established proportional rules as deserving of representation.}
In contrast, Sainte-Laguë Approval Voting, which plays a prominent role in our paper, treats both groups almost equally ($4.9$ for group~2).

Depending on the application, avoiding overrepresentation can be as important as guaranteeing representation (or more). In parliamentary apportionment~\citep{BaYo01a,Puke14a}, a special case of multi-winner voting, avoiding overrepresentation is equivalent to what is known as \emph{respecting upper quota} (or \emph{upper ideal frame}). This notion has received much attention and can be seen as a way to favor small parties. Going in a more technical direction,~\citet{CeSt21a} highlight the connection between overrepresentation and security in distributed systems, particularly blockchain networks. More specifically,~\citet{BBC+24a} study overrepresentation in the context of \emph{Polkadot}, a blockchain with a mechanism based on multi-winner voting, again highlighting the link between overrepresentation and security.\footnote{However, note that the rules considered by~\citeauthor{CeSt21a}, maximin-support and Phragmén-maximin-support, later also considered by \citeauthor{BBC+24a}, do not respect upper quota even in the apportionment setting. Indeed, maximin-support was originally introduced by~\citet{SFFB24a} as a rule that generalizes D'Hondt's method of apportionment, which respects lower quota but not upper quota.} Beyond multi-winner voting, \citet{jerrett2025low} recently studied the related problem of selecting a representative subset from a set of agents embedded in a metric space. They argue that solutions that overrepresent certain groups are undesirable and propose an axiom for avoiding them.

Despite the importance of this matter, previous work in approval-based multi-winner voting has mostly focused on avoiding underrepresentation. Indeed, rules that respect upper quota have received almost no attention. For example, it is known that PAV is the only rule in the class of \emph{Thiele rules}, a normatively attractive family of rules, that satisfies strong guarantees against underrepresentation (such as extended justified representation~\citep{ABC+16a,SFF+17a}). No corresponding results are known for overrepresentation. Notably, in their survey book on approval-based multi-winner voting, \citet{LaSk22a} include the study of proportionality beyond underrepresentation as an open question. They call for further analysis of rules that avoid overrepresentation, as well as rules that generalize the Sainte-Laguë method of apportionment (which can be seen as a way to balance under- and overrepresentation) to the multi-winner setting. However, even defining a reasonable overrepresentation axiom for the multi-winner setting poses some conceptual difficulties (which we explore in Section~\ref{sec:JUQ}).

Our contributions approach overrepresentation from several angles:
\begin{itemize}
	\item We introduce a strong and appealing upper quota axiom, called \emph{justified upper quota (JUQ)}. While this axiom is not satisfiable within the well-known class of Thiele rules, it can be satisfied in the class of \emph{composite Thiele rules}. The latter is analogous to the class of \emph{composite social choice scoring functions} in standard single-winner voting~\citep{Youn75a}. Intuitively, composite Thiele rules are defined through a sequence of Thiele rules and work in stages. First, we compute the result of the first Thiele rule. Then, at each subsequent stage, we apply the next rule to break the ties of the previous stage. To the best of our knowledge, this is the first paper to introduce composite Thiele rules. 
	
	\item We show that, among this general class, JUQ pinpoints a unique rule: \emph{Adams-AV}. This is the natural extension to the multi-winner voting setting of an apportionment rule proposed by John Quincy Adams. The latter is notable as the only population monotone
	apportionment rule that respects upper quota~\citep{BaYo01a}, which, as discussed above, corresponds to avoiding overrepresentation in the apportionment setting. To the best of our knowledge, although~\citet{LaSk21a} briefly mentioned this rule, we are the first to initiate the axiomatic study of Adams-AV. Finally, since computing Adams-AV corresponds to an NP-hard optimization problem, we also introduce a polynomial-time rule satisfying JUQ.
	
	\item Taking inspiration from the near quota axiom in the apportionment literature, we also introduce \emph{justified near quota} (JNQ), a notion that captures both under- and overrepresentation by requiring each group's representation to deviate as little as possible from its quota. We show that Sainte-Laguë Approval Voting, a Thiele method based on the Sainte-Laguë method of apportionment (which satisfies the near quota axiom~\citep{BaYo01a} in that setting), satisfies JNQ and is in fact the only composite Thiele rule with this property.

	\item Finally, we investigate the compatibility of upper, near, and lower quota axioms. We show that if we allow committees to be smaller than the target committee size $k$, we can guarantee a stronger version of JUQ, a stronger version of extended justified representation (EJR+, introduced by \citet{BrPe23a}) and JNQ simultaneously. We identify further conditions under which these axioms can be guaranteed jointly. Whether any two of JUQ, JNQ, and EJR+ are compatible for committees of size exactly $k$ remains an intriguing open problem, and 
	we show that some natural approaches to proving this compatibility cannot succeed.
\end{itemize}

This is an extended version of a paper presented earlier at a workshop \citep{overrepresentation-mpref}.
All omitted proofs are deferred to Appendix~\ref{app:omitted_proofs}.

\section{Model and Voting Rules}\label{sec:preliminaries}

In this section, we introduce the model of \emph{approval-based multi-winner voting}~\citep{LaSk22a} and its special case, the \emph{apportionment} problem~\citep{BaYo01a,Puke14a}.

We denote by $\nats$ the set of natural numbers including $0$, i.e., the set $\{0, 1, 2, \ldots\}$. Given a positive integer $\ell\in\posNats$, we define $[\ell]=\{1,\,\ldots,\,\ell\}$. Furthermore, given a set $S$, we denote by $\powerset{S}$ its powerset and for $\ell \in \nats$ by $\npowerset{\ell}{S}=\{T\in\powerset{S}\colon |T|=\ell\}$ the set of subsets of size exactly $\ell$.

\subsection{Approval-Based Multi-Winner Voting}

Let $\voters=[n]$ be the set of \emph{voters} and let $\candidates$ be the set of \emph{candidates} with $|\candidates|=m>0$. An \emph{approval profile} (or simply \emph{profile}) is a tuple $\approfile=(\app_1,\,\ldots,\,\app_n)\in\powerset{\candidates}^n$, where $A_i\subseteq\candidates$ for every $i\in N$. 
For a given voter $i \in N$, we call $A_i$ the \emph{approval ballot} of this voter. Given a candidate $c\in\candidates$, we denote its set of \emph{supporters} by
\(\voters(c)=\{i\in \voters\colon c\in \app_i\}\)
and its size by $n_c=|\voters(c)|$. For simplicity we make the (technical) assumption that $n_c > 0$ for all $c \in C$.

Further, we are given a \emph{target committee size} $k \in [m]$. As every candidate is approved by at least one voter, we can uniquely identify an election by the approval profile $\approfile$ together with the committee size $k$ (by taking $C = \bigcup_{i \in \voters} A_i$). We call $(\approfile, k)$ an \emph{instance}. A \emph{voting rule} is a function $F$ that maps every instance $(\approfile, k)$ to a non-empty set of committees $\mathcal{W}\subseteq\{W\subseteq\candidates\colon|W|\leq k\}$. Observe that we allow rules to return committees smaller than the target size; we say that $F$ is \emph{exhaustive} if it always returns committees of size $k$, i.e., if $F(\approfile, k)\subseteq\kcommittees$ for all $(\approfile, k)$.\footnote{Note that exhaustiveness makes avoiding overrepresentation more difficult. Intuitively, the empty committee is not overrepresenting any group of voters (as no voter is represented at all).} Given two voting rules $F$ and $G$, we say that $F$ is a \emph{refinement} of $G$ if $F(\approfile,k)\subseteq G(\approfile,k)$ for all inputs.

For a set of voters $\voters^\prime\subseteq\voters$ and a target committee size $k$, we define its \emph{quota} as
\( q(\voters^\prime)= k\frac{|\voters^\prime|}{|\voters|}. \)
We also define $q_c=q(\voters(c))$ for a candidate $c \in C$. The quota of a group represents the (possibly fractional) number of seats on the committee that corresponds to the size of the group. Intuitively, $\voters^\prime$ would be entitled to fill this many seats on the committee if they can agree on sufficiently many candidates to do so. 
We call $\lceil q(\voters^\prime)\rceil$ the \emph{upper quota} and $\lfloor q(\voters^\prime)\rfloor$ the \emph{lower quota} of $\voters^\prime$. In the next subsection, we present an example illustrating quotas (\Cref{tab:apportionment-instance}).

Ideally, we would like to guarantee that, for each group $\voters^\prime$, the committee contains at least $\lfloor q(\voters^\prime)\rfloor$ candidates approved by all voters in $\voters^\prime$. As it is easy to see that this is impossible in general \citep[see, e.g.,][]{BIMP24a}, several weaker lower quota axioms have been proposed in the literature~\citep{ABC+16a, SFF+17a, PPS21a, BrPe23a, KKL25a}. 
We will focus on \emph{extended justified representation plus} (EJR+), due to \citet{BrPe23a}, a strong proportional representation axiom satisfiable in polynomial time. We express it in terms of quota---differently than how it was initially introduced---to match our notation.
\begin{definition}
	Given a target committee size~$k$, a set $W \subseteq C$ satisfies \emph{EJR+} if there is no candidate $c\in C\setminus W$ and non-empty group of voters $N'\subseteq N(c)$ such that
	$ |A_i\cap W|<\floor{q(N')} \text{ for all } i\in N'.$
\end{definition}
That is, there should be no group $N'$ of voters approving a common unchosen candidate such that everyone in this group approves fewer committee members than the group's lower quota.

\subsection{Apportionment}

The \emph{apportionment} setting is a special case of the approval-based multi-winner voting model.
Intuitively, in apportionment, candidates are partitioned into parties, and each voter approves all candidates of exactly one party; in this sense, voters are partitioned into parties as well.
Formally, an instance $(\approfile, k)$ is an \emph{apportionment instance}, if for every $i,\,j\in\voters$ we have either $A_i=A_j$ or $A_i\cap A_j=\emptyset$. Furthermore, for every $i\in\voters$ we must have $|A_i|=k$, i.e., we assume that each party has enough candidates to fill every seat. An \emph{apportionment rule} is an exhaustive voting rule whose domain is restricted to apportionment instances. We say that a voting rule $F$ \emph{extends} an apportionment rule $G$ if $F(\approfile, k) = G(\approfile, k)$ for all apportionment instances $(\approfile, k)$. 

We introduce some additional notation for apportionment instances. We partition $\voters$ into $t$ disjoint equivalence classes called \emph{parties}, that is, $N=P_1\cup\cdots\cup P_t$. Two voters have the same ballot if and only if they belong to the same party. Given a committee $W$, we define the \emph{seats} given to party $P_\ell$ as $a_\ell(W)=|A_i\cap W|$ (for some arbitrary $i\in P_\ell$). When $W$ is fixed, we just write $a_\ell$.

\begin{definition}
An apportionment rule satisfies \emph{apportionment upper quota} (resp., \emph{lower quota}) if, for every outcome $W$ selected by this rule and party $P_\ell$, we have that $a_\ell\leq\lceil q(P_\ell)\rceil$ (resp., $a_\ell\geq\lfloor q(P_\ell)\rfloor$). Moreover, an apportionment rule satisfies \emph{apportionment near quota} if for every outcome $W$ selected by this rule there do not exist two parties $P_i$ and $P_j$ such that 
\(
\lvert q(P_i) - (a_i - 1)\rvert < \lvert q(P_i) - a_i\rvert \text{ and }  \lvert q(P_j) - (a_j + 1)\rvert < \lvert q(P_j) - a_j\rvert.
\)
In other words, we should not be able to take a seat from party $P_i$ and allocate it to $P_j$ while bringing both parties closer to their respective quotas.
\end{definition}

We refer the reader to the book of~\citet{BaYo01a} for an extended discussion of lower, upper, and near quota in the apportionment setting.

\begin{example}
	\begin{wrapstuff}[r,type=table,width=6.cm,abovesep=0cm,belowsep=-0.2cm]
		\centering
		\begin{tabular}{lcccc}
			\toprule
			& $P_1$ & $P_2$ & $P_3$ & $P_4$ \\
			\midrule
			$\lvert P_i\rvert$ & 24 & 36 & 35 & 15 \\
			$q(P_i) = \nicefrac{k|P_i|}{n}$ & 2.4 & 3.6 & 3.5 & 1.5 \\
			$\lfloor q(P_i)\rfloor$ & 2 & 3 & 3 & 1 \\
			$\lceil q(P_i)\rceil$ & 3 & 4 & 4 & 2 \\
			\midrule
			$a^{(1)}_i$ & 3 & 3 & 3 & 2 \\
			$a^{(2)}_i$ & 1 & 4 & 4 & 2 \\
			$a^{(3)}_i$ & 2 & 5 & 3 & 1 \\
			$a^{(4)}_i$ & 2 & 4 & 4 & 1 \\
			\bottomrule
		\end{tabular}
		\caption{Apportionment instance depicting differences in quota concepts ($n=110$ and $k=11$).}
		\label{tab:apportionment-instance}
	\end{wrapstuff}
	 To illustrate the three different quota concepts, consider the apportionment instance depicted in \Cref{tab:apportionment-instance}. In the instance we abstract from the underlying voters and candidates and consider each party only by its vote count, and each outcome only by its seat allocation. For the instance, we assume $k = 11$ and have $n = 110$. First, allocation $a^{(1)}$ satisfies lower and upper quota. It does not, however, satisfy near quota, as reallocating one seat from party $P_1$ to party $P_2$ would decrease their respective distances to their quotas by $0.2$ each. Allocation $a^{(2)}$ satisfies upper quota and near quota. It does not, however, satisfy lower quota due to party $P_1$. Allocation $a^{(3)}$ satisfies lower quota and near quota, but not upper quota, due to party $P_2$. Finally, allocation $a^{(4)}$ satisfies all three.\qedhere
	\wrapstuffclear
\end{example}

It is easy to see that for apportionment instances, EJR+ and lower quota coincide. 
\begin{restatable}{observation}{obsejrpeq}
	For apportionment instances, EJR+ is equivalent to apportionment lower quota.
\end{restatable}
As a final ingredient of the apportionment setting, we introduce \emph{divisor methods}, a normatively appealing class of apportionment rules~\citep{BaYo01a}.\footnote{Divisor methods are the only apportionment rules that are \emph{population monotone}. Population monotonicity captures the idea that if a party $P_i$'s size increases while $P_j$'s decreases, then $a_i$ should not decrease while $a_j$ increases. \citet{BaYo01a} define divisor methods in a slightly different but equivalent way (see Proposition 3.3 of their book).}

\begin{definition}
	A divisor method is defined by a function $d\colon\nats\rightarrow\nonNegReals$ with $d(i) \le d(i + 1)$ for every $i \in \nats$. We sometimes write $d$ as a vector. This method proceeds sequentially, starting with $W = \emptyset$ and iteratively allocating the seat to a party $P_\ell$ maximizing $\nicefrac{\lvert P_\ell\rvert}{d(a_{\ell})}$ by adding an arbitrary unchosen candidate approved by this party to $W$. This is repeated until $\lvert W\rvert = k$, and the method returns all committees that can be obtained by the above procedure plus \emph{some} way of breaking ties. We adopt the convention that $\nicefrac{i}{0}=\infty$ for all $i\in\posNats$ and that $\nicefrac{i}{0} > \nicefrac{j}{0}$ if $i > j$.
\end{definition}

We highlight three well-known divisor methods which are relevant to our analysis:
\begin{itemize}
	\item \emph{D'Hondt's} (or \emph{Jefferson's}) \emph{method} is given by $d_H=(1,\,2,\,3,\,\ldots)$;
	\item \emph{Adams' method} is given by $d_A=(0,\,1,\,2,\,\ldots)$;
	\item \emph{Sainte-Laguë's} (or \emph{Webster's}) \emph{method} is given by $d_S=(1,\,3,\,5,\,\ldots)$.
\end{itemize}

These three methods are notable as they are the only divisor methods that satisfy, respectively, lower, upper, and near quota in apportionment~\citep{BaYo01a}. Note that, since $d_A(0)=0$ and $\nicefrac{s}{0}=\infty$ for all $s\in\posNats$, Adams' method starts by assigning one seat to each party ordered by their size (until all parties are assigned at least one seat or the committee size $k$ is reached) and then proceeds similarly to D'Hondt's method.

\subsection{Multi-Winner Voting Rules}

Since Adams' and Sainte-Laguë's methods satisfy upper quota and near quota, respectively, in apportionment,
generalizing them to multi-winner voting is a natural first step towards understanding overrepresentation
in this setting.
Most divisor methods have direct counterparts in the popular class of \emph{Thiele rules}, introduced by Thorvald N.\ Thiele~\citep{Thie95a,Jans16a}. These (exhaustive) rules are well-studied in multi-winner voting~\citep{LaSk22a}. Intuitively, each Thiele rule is defined by a \emph{weight function} which determines the score that a voter assigns to a committee given the number of candidates they approve in that committee.

However, Adams' method does not correspond directly to a Thiele rule, as we will discuss later. 
Thus, we introduce a more general class of rules, called \emph{composite Thiele rules}. A composite Thiele rule is defined by a sequence of Thiele rules. First, we select a set of winning committees via the first rule. Then, we select a subset of these committees through the second rule. This process continues until either every rule in the sequence has been applied, or we reach a fixed point. 

\begin{definition}
	A \emph{weight function} is a function $w\colon\posNats\rightarrow\nonNegReals$ such that for every $s\in\posNats$, $w(s)\geq w(s+1)$. We often write weight functions as vectors and normalize $w(1)=1$.
  Given $\approfile$, $k$, and $i\in\voters$, we define
	\(\score_{w}(\app_i, W) = \sum_{s=1}^{|\app_i\cap W|}w(s)\text{ and }\score_{w}(\approfile, W)=\sum_{j\in \voters}\score_{w}(\app_j, W).\)
  The \emph{simple} or \emph{1-composite Thiele rule} corresponding to $w$ is then
  \(F_w(\approfile,k) = \argmax_{W\in\kcommittees}\ \ \score_{w}(\approfile,W). \) \sloppy
  Given two weight functions $w_1$ and $w_2$, we define the composition
  \( \left(F_{w_2}\circ F_{w_1}\right)(\approfile,k) = \argmax_{W\in F_{w_1}(\approfile,k)}\ \score_{w_2}(\approfile, W). \)
  Then, finally, for any sequence $\boldsymbol{w}=(w_1,w_2,\ldots)$ of length $\ell\in\posNats\cup\{\infty\}$ of weight functions, the corresponding \emph{$\ell$-composite Thiele rule} is
  \(F_{\boldsymbol{w}}(\approfile,k) = \bigcap_{i=1}^\ell (F_{w_i}\circ \cdots \circ F_{w_1})(\approfile,k). \)
  Note that if $\boldsymbol{w}=(w_1,\ldots,w_\ell)$ is finite the above collapses to $F_{\boldsymbol{w}}(\approfile,k) = \left(F_{w_\ell}\circ \cdots \circ F_{w_1}\right)(\approfile,k)$ and that this is well-defined even for $\ell = \infty$ as there are only finitely many committees.
\end{definition}

We allow for $\infty$-composite Thiele rules to also capture rules such as the one given by the infinite sequence $(1,\,0,\,0,\,\ldots)$, $(1,\,1,\,0,\,0,\,\ldots)$, $(1,\,1,\,1,\,0,\,0,\,\ldots)$, etc. This is an analogue to the leximin social welfare function~\citep{Rawl71a,Sen17a}.\footnote{Observe that, since we assume $n$ to be fixed, we do not technically need infinite sequences to capture this rule. However, we stick with this definition (which accommodates for electorates of variable size) as our results also hold for this more general case.}

To illustrate our definition further, we present some standard Thiele voting rules:
  \begin{itemize}
    \item \emph{Approval Voting} (AV) is the Thiele rule given by $w=(1,\,1,\,1,\,\ldots)$.
    \item \emph{Chamberlin-Courant} (CC) is the Thiele rule given by $w=(1,\,0,\,0,\,\ldots)$.
    \item \emph{Proportional AV} (PAV) is the Thiele rule given by $w=(1,\,\nicefrac{1}{2},\,\nicefrac{1}{3},\,\ldots)$.
  \end{itemize}

Let us give an intuitive idea of these rules. AV picks the committees with the highest total approval score. On the other hand, CC selects the committees that cover as many voters as possible, where a voter is covered if they approve at least one candidate in the committee (the score of a committee does not increase if a voter approves more than one candidate). One can think of AV and CC as lying on two different ends of the spectrum of (simple) Thiele rules: while AV is entirely utilitarian, CC tries to satisfy all voters whenever possible. PAV lies, in some sense, in between these two. Given a committee $W$, under PAV each voter assigns to $W$ a score corresponding to the $|\app_i\cap W|$-th harmonic number; the committees with the highest total score are selected. Thus any additional candidate a voter approves in the committee will yield diminishing returns. 

There is a close connection between divisor methods and Thiele rules: the divisor method given by $d$ is extended uniquely among simple Thiele rules by the one defined by $w(s)=\nicefrac{d(0)}{d(s-1)}$~\citep{LaSk21a,BLS18a}. Note that this is not well-defined for divisor methods with $d(0) = 0$, such as Adams' method. We note that PAV uniquely extends D'Hondt's method and is the only Thiele rule that satisfies EJR+~\citep{ABC+16a,BrPe23a}. Since we are interested in upper and near quota notions, it is then natural to look at the Thiele rules that extend Adams' and Sainte-Laguë's methods.
For Sainte-Laguë, the corresponding Thiele rule is:

\begin{definition}
	\emph{Sainte-Laguë-AV} (SLAV) is the simple Thiele rule given by $w=(1,\,\nicefrac{1}{3},\,\nicefrac{1}{5},\,\ldots)$.
\end{definition}

In contrast, Adams' method does not have a corresponding Thiele method, as its divisor sequence starts with $0$ and hence would imply $w(1)=\infty$. This behavior, however, can be captured in the class of composite Thiele rules.

\begin{definition}
  \emph{Adams-AV} is the $2$-composite Thiele rule given by $w_1=(1,\,0,\,0,\,\ldots)$ and $w_2=(1,\,1,\,\nicefrac{1}{2},\,\nicefrac{1}{3},\,\ldots)$.\label{def:adams-av}
\end{definition}

Observe that Adams-AV first maximizes the number of voters with at least one candidate in the committee (just as CC does), then proceeds with a PAV-like weight function. This resembles exactly the behavior of Adams' method and indeed, we can verify that Adams-AV extends Adams' method.
\begin{restatable}{observation}{adamsextends}
	Adams-AV extends Adams' method.
\end{restatable}

Returning to the example from the introduction (Figure~\ref{fig:unbalanced}), Adams-AV strongly advantages the diverse group (group~2) with 7.5 chosen candidates on average.\footnote{One may wonder whether this is actually an \emph{overrepresentation} of group~2, and indeed, this is a fair criticism. Adams-AV considers the small diverse groups within group~2 as distinct in the same way as PAV considers them distinct to justify ignoring them. Consequently, following the logic of this example, group~1 and~2 are only treated equally if both under- and overrepresentation are avoided.}
This further highlights the connection to CC, which is known to prioritize \emph{diversity} (i.e., it is representation-focused) \citep{FSST17a,LaSk20b}.
To the best of our knowledge, Adams-AV has not been studied previously in the literature.

\begin{example}
	Consider profile $\approfile=(\{a,b,c,d\},\,\{a,b,c,d\},\,\{a,b,c,e\},\,\{d\},\,\{b,c,e\},\,\{c,e\})$ and let $k=2$. Here, AV picks only $\{b,\,c\}$, the two most approved candidates. CC selects $\{c,\,d\}$ and $\{d,\,e\}$, as they both cover all voters. PAV selects $\{b,\,c\}$ and $\{c,\,d\}$. Among the two CC-optimal committees, Adams-AV picks only committee $\{c,\,d\}$. Similarly, SLAV picks only committee $\{c,\,d\}$.
\end{example}

\section{Avoiding Overrepresentation}\label{sec:JUQ}

In the previous section, we introduced Adams-AV, a multi-winner rule that extends Adams' apportionment method. However, it is not immediately clear whether this rule always avoids overrepresentation, or even how we should extend the idea of overrepresentation to multi-winner elections. In this section, we answer both of these questions by generalizing the normative notion of an upper quota of parties to groups of voters in the approval-based multi-winner setting.

\subsection{Justified Upper Quota}

Let us begin by defining a natural upper quota analog of EJR+ by requiring that there is no selected candidate whose supporters all receive more than this candidate's upper quota. Note that, in contrast to EJR+, we only consider the full set of supporters $\voters(c)$ and not its subsets, as there could exist small subsets that are overrepresented while the entire set of approvers is underrepresented.

\begin{definition}
	A set $W \subseteq C$ of candidates satisfies \emph{upper quota} (UQ) if there is no selected candidate $c\in W$ with 
	$|A_i\cap W|>\ceil{q_c}$ for all $i\in \voters(c).$
\end{definition}
Note that to be overrepresented \emph{everyone} needs to be over the quota, similarly to how EJR-style axioms require all members of a particular voter set to be below quota for underrepresentation to occur.
UQ is trivially satisfiable by non-exhaustive committees, in particular by the empty committee.\footnote{Note that $q_c$ is defined in terms of the target committee size $k$, not the actual size of $W$.} However, it can be infeasible for exhaustive committees.

\begin{example}\label{ex:why-difficult}
	Consider the profile $\approfile=\left(\{a,b,c\},\,\{d,e,f\},\,\{g\}\right)$ and let $k=6$. This is \emph{almost} an apportionment instance, but here voters approve fewer than $k$ candidates. Any exhaustive committee must contain either three candidates approved by voter~$1$ or three candidates approved by voter~$2$ (or both). Suppose, without loss of generality, that $\{a,b,c\}\subseteq W$. Then we have that $|A_1\cap W|=3>\ceil{q_c}=2$, violating UQ. A similar result holds for any other exhaustive committee.
\end{example}

Due to the closeness of Example~\ref{ex:why-difficult} to the apportionment setting, it seems hard to argue with the interpretation that voter 1 and 2 form a unique party or group each and that, if we select three candidates they approve, they are overrepresented. Thus, we must conclude that in some cases overrepresentation is unavoidable for exhaustive committees. However, we only want to permit overrepresentation if it is---in some sense---necessary. We interpret this necessity as there being no other group that could receive a candidate without any of the voters breaching their upper quota. We call such violations of upper quota ``justified'' and introduce the axiom \emph{justified upper quota}.

\begin{definition}\label{def:JUQ}
	Committee $W$ satisfies \emph{justified upper quota} (JUQ) if, whenever there is a $c\in W$ such that
	\[ |\app_i\cap W| > \lceil q_c\rceil\quad\text{for all }i\in \voters(c),\]
	then there exists no $d\in \candidates\setminus W$ and nonempty $\voters^\prime\subseteq\voters(d)$ such that
	\[ |\app_i\cap ((W\cup\{d\})\setminus\{c\})|\leq\lceil q(\voters^\prime)\rceil\quad\text{for all }i\in \voters^\prime.\qedhere\]
\end{definition}
Note that we allow for subsets $N' \subseteq N(d)$ in our definition: for us it is sufficient that there exists some non-empty subgroup of the approvers of $d$ who could be awarded an extra candidate without breaching their upper quota. 
Observe that UQ implies JUQ, and that we require $\voters^\prime\neq\emptyset$ as otherwise the second condition would be vacuously true for any unselected candidate.

Even when every exhaustive committee overrepresents some group according to UQ, JUQ can remain meaningfully satisfiable.

\begin{example}
	Let us revisit the scenario of Example~\ref{ex:why-difficult}, where $\approfile=(\{a,b,c\},\,\{d,e,f\},\,\{g\})$ and $k=6$. First, consider the committee $W_1=\{a,\,b,\,c,\,d,\,e,\,g\}$. This committee satisfies JUQ. Indeed, consider, for example, $c\in W_1$. We have that $\lceil q_c\rceil=2$, but voter~$1$ has satisfaction $3$, that is, she approves $3$ candidates in the committee. However, the only candidate not in $W_1$ is $f$, and we have that, if we were to replace $c$ with $f$ in $W_1$, then voter~$2$ would get a satisfaction of $3>2=\lceil q_f\rceil$. On the other hand, the committee $W_2=\{a,\,b,\,c,\,d,\,e,\,f\}$ violates JUQ. Indeed, again $c$ is a witness for an upper quota violation. However, now we can swap $c$ with $g$, as $\lceil q_g\rceil=2\ge\lvert A_3 \cap ((W_2 \cup \{g\}) \setminus \{c\})\rvert$.
\end{example}

Note that JUQ only disallows overrepresentation in the most egregious cases: if \emph{all} members in a group $\voters(c)$ are above their upper quota, and there is another group $\voters^\prime\subseteq\voters(d)$ where \emph{all} members would be at most at their upper quota after replacing $c$ with $d$, then JUQ is violated. Note the asymmetry: it would perhaps be more natural to require only \emph{some} members of $\voters^\prime$ to be at most their upper quota to warrant the replacement of $c$ by $d$. However, there are reasons to believe that this \emph{symmetric variant} of JUQ is too strong, as we will discuss at the end of Section \ref{sec:char-adams}. We will therefore focus on JUQ and leave a more in-depth analysis of this stronger notion to future work.

Next, similarly to the relationship between EJR+ and apportionment lower quota, we can show that JUQ and apportionment upper quota are equivalent on apportionment instances.
\begin{proposition}\label{prop:JUQ-apportionment}
	For apportionment instances, JUQ is equivalent to apportionment upper quota.
\end{proposition}

\begin{proof}
	Consider an apportionment instance $(\approfile, k)$ and consider any committee $W$ of size at most $k$. First, assume that $W$ violates JUQ. Then there exists a candidate $c \in W$ (belonging to a party $P_\ell$)  such that $\lceil q(P_\ell)\rceil = \lceil q_c \rceil < \lvert A_i \cap W\rvert = a_\ell $ and hence $W$ also violates apportionment upper quota. 	
	For the other direction, we observe that by definition of the quota it holds that 
	\(
	\sum_{\ell = 1}^t a_{\ell} \le k = \sum_{\ell = 1}^t k \cdot \frac{\lvert P_\ell \rvert}{n} = \sum_{\ell = 1}^t q(P_\ell).	\)
	Hence, if there exists a party $P_\ell$ with $a_{\ell} > q(P_\ell)$ there must exist a second party $P_s$ with $a_s < q(P_s)$ and thus as $a_s$ is an integer $a_s + 1 \le \lceil q(P_s)\rceil$. Hence, if apportionment upper quota is violated, we can take any selected candidate from party $P_\ell$ as $c$, any unselected candidate from party $P_s$ as $d$ and $P_s$ as $N'$ as a witness that JUQ is also violated.
\end{proof}

We remark that JUQ is verifiable in polynomial time (as is the case for EJR+~\citep{BrPe23a}, but not the original EJR notion~\citep{ABC+16a}). 

\begin{restatable}{proposition}{JUQVerifPoly}
	Whether a committee satisfies JUQ is decidable in polynomial time.\label{prop:juq-verif-poly}
\end{restatable}

\subsection{Characterization of Adams-AV}\label{sec:char-adams}
In the previous section, we introduced JUQ as our upper quota axiom. It is easy to see that we can always satisfy JUQ via a non-exhaustive committee: for instance, the empty committee trivially satisfies it. However, similar to UQ, it is not obvious that there always exists a committee of size $k$ satisfying JUQ. We show that this is always the case. In particular, we prove that every committee selected by Adams-AV satisfies JUQ.

\begin{theorem}\label{thm:JUQ-adams}
	Any refinement of Adams-AV satisfies justified upper quota.
\end{theorem}

\begin{proof}
	Let $W$ be a committee of size $k$ not satisfying JUQ. Our goal is to show that $W$ cannot be selected by Adams-AV. Let $w_1=(1,0,0,\ldots)$ and $w_2=(1,1,\nicefrac{1}{2},\nicefrac{1}{3},\ldots)$ be the two weight vectors used to define Adams-AV. Assume that there is a JUQ violation witnessed by the two candidates $c \in W$ and $d \in C \setminus W$, as well as by $N' \subseteq N(d)$. We consider the committee $W' = (W \setminus \{c\}) \cup \{d\}$, swapping the two candidates.	
	Hence, our goal becomes to show that either
	 \begin{enumerate}
		\item $\score_{w_1}(\approfile, W')  > \score_{w_1}(\approfile, W)$ or \item $\score_{w_2}(\approfile, W') > \score_{w_2}(\approfile, W)$ and $\score_{w_1}(\approfile, W') =\score_{w_1}(\approfile, W)$. 
	\end{enumerate}  Hence, performing the swap increases the ``score'' associated with Adams-AV and the committee $W$ could not have been selected by it.
	
	 Since $c,d$, and $N'$ witness the JUQ violation it holds that $\lvert A_i \cap W\rvert > \ceil{q_c}$ for all $i \in N(c)$ and $\lvert A_i \cap W' \rvert \le \ceil{q(N')}$ for all $i \in N'$. First, we notice that since $\lvert A_i \cap W\rvert > \ceil{q_c}$ it must particularly hold that $\lvert A_i \cap W\rvert > 1$ for all $i \in N(c)$. After removing $c$ from $W$ every voter in $\voters(c)$ still approves at least one candidate of  $W'$. Thus, if there exists a voter $i \in N(d)$ with $A_i \cap W = \emptyset$ the $\score_{w_1}$ would strictly increase as we cover more voters. If, however, $A_i \cap W \neq \emptyset$ for all $i \in N(d)$, we would get $\score_{w_1}(\approfile, W') = \score_{w_1}(\approfile, W)$.
	
	Thus, for the second part we can assume that  $A_i \cap W \neq \emptyset$ for all $i \in \voters(d)$ (and thus also for all $i \in N'$). Now, our goal is to lower bound the score difference $	\score_{w_2}(\approfile, W') - \score_{w_2}(\approfile, W)$ and to show that this difference is always positive. First, we notice that for any $i \in N(c) \cap N(d)$ or $i \in N \setminus (N(c) \cup N(d))$ the number of approved candidates is the same in $W'$ and $W$. Hence, for such a voter  $	\score_{w_2}(A_i, W') - \score_{w_2}(A_i, W) = 0$. We note that for the $w_2$ score vector, for a voter with $\lvert A_i \cap W\rvert \ge 1$ approved candidates on the committee, the marginal increase of adding another approved candidate is $\frac{1}{\lvert A_i \cap W\rvert}$, while the marginal loss of removing one is $\frac{1}{\lvert A_i \cap W\rvert - 1} $. For any $i \in N(d) \setminus N(c)$ we know by our assumption that $i$ approves at least one candidate in $W$. Hence, the score increases by $\score_{w_2}(A_i, W') - \score_{w_2}(A_i, W) = \frac{1}{\lvert A_i \cap W\rvert}$. Similarly, for any $i \in N(c) \setminus N(d)$ the score must decrease by $	\score_{w_2}(A_i, W') - \score_{w_2}(A_i, W) = - \frac{1}{\lvert A_i \cap W\rvert - 1}$. We get that 
	\begin{align*}
		\score_{w_2}(\approfile, W') - \score_{w_2}(\approfile, W) &= \sum_{i \in N(d) \setminus N(c)} \frac{1}{\lvert A_i \cap W\rvert}-\sum_{i \in N(c) \setminus N(d)} \frac{1}{\lvert A_i \cap W\rvert - 1} \\
		&\ge  \sum_{i \in N' \setminus N(c)} \frac{1}{\lvert A_i \cap W\rvert}-\sum_{i \in N(c) \setminus N'} \frac{1}{\lvert A_i \cap W\rvert - 1}.
	\end{align*}
	For any $i \in N' \setminus N(c)$, since $N'$ and $d$ witness the JUQ violation, we know that $\lvert A_i \cap W \rvert  + 1 \le \lceil q(N')\rceil$ and therefore $\lvert A_i \cap W \rvert < q(N')$. Similarly, for any $i \in N(c) \setminus N'$ it must hold that $\lvert A_i \cap W\rvert \ge q_c + 1$. Further, we note that $N'\setminus N(c)$ must always be non-empty as otherwise $N'\subseteq N(c)$ and thus $q(N') \le q_c$. However, then for any $i \in N'$ we would get $\lceil q(N')\rceil  \le \lceil q_c\rceil < \lvert A_i \cap W\rvert$, a contradiction to the second condition of JUQ.	
	 Thus, we obtain
	\begin{align*}
	\sum_{i \in N' \setminus N(c)} \frac{1}{\lvert A_i \cap W\rvert}-\sum_{i \in N(c) \setminus N'} \frac{1}{\lvert A_i \cap W\rvert - 1} &> \sum_{i \in N'\setminus N(c)} \frac{1}{q(N')}-\sum_{i \in N(c) \setminus N'} \frac{1}{ q_c}\\ = \frac{\lvert N'\setminus N(c)\rvert}{q(N')} - \frac{\lvert N(c) \setminus N'\rvert}{ q_c} & = \frac{\lvert N'\rvert}{q(N')} - \frac{\lvert \voters(c)\rvert}{ q_c} - \frac{\lvert N' \cap N(c)\rvert }{q(N')} + \frac{\lvert N' \cap N(c)\rvert }{q_c} .
	\end{align*}
	Now, by definition of the quota $\frac{\lvert N'\rvert}{q(N')} = \frac{n}{k}$ and $ \frac{\lvert \voters(c)\rvert}{ q_c} = \frac{n}{k}$. Hence, this expression simplifies to 
	\begin{align*}
		\frac{\lvert N'\rvert}{q(N')} - \frac{\lvert \voters(c)\rvert}{ q_c} - \frac{\lvert N' \cap N(c)\rvert }{q(N')} + \frac{\lvert N' \cap N(c)\rvert }{q_c} = \lvert N' \cap N(c)\rvert \left(\frac{1}{q_c}-\frac{1}{q(N')}\right).
	\end{align*}
	If $N' \cap N(c) = \emptyset$ this is $0$ and hence  $\score_{w_2}(\approfile, W') - \score_{w_2}(\approfile, W) > 0$ would hold. Otherwise, there must exist a voter $j \in N' \cap N(c)$. For this voter, we know that $\lceil q_c \rceil < \lvert A_j \cap W\rvert = \lvert A_j \cap W'\rvert \le \lceil q(N') \rceil$. Thus, $q_c < q(N')$ and $(\frac{1}{q_c} - \frac{1}{q(N')})$ must be positive.
	Consequently, $\score_{w_2}(\approfile, W') - \score_{w_2}(\approfile, W) > 0$ also holds and $W$ could not have been selected by Adams-AV.
\end{proof}

We complement this by a characterization of Adams-AV among composite Thiele rules: A composite Thiele rule satisfies apportionment upper quota if and only if it is a refinement of Adams-AV. 
\begin{restatable}{theorem}{adamschar}\label{thm:AUQ-thiele}
	A composite Thiele rule satisfies apportionment upper quota if and only if it is a refinement of Adams-AV.
\end{restatable}

Note that \Cref{thm:AUQ-thiele} can \emph{in principle} be derived from the theory of apportionment, at least when it comes to a class of simple Thiele rules with strictly positive weights. Indeed, we know that every such simple Thiele rule extends a unique divisor method~\citep{LaSk21a,BLS18a}, and that all divisor methods except Adams' violate apportionment upper quota~\citep{BaYo01a} (at least in the definition of \citet{BaYo01a} which forces $d(i) \in [i, i+1]$---with appropriate normalization---for any divisor method with function $d$). However, for non-simple Thiele rules and those not fitting the definition of \citeauthor{BaYo01a}, this does not follow. We thus give a direct proof of the result (in Appendix~\ref{app:omitted_proofs}).

We further observe that Adams-AV fails the stronger symmetric variant of JUQ discussed in the previous section. Consider an $8$-voter profile $\approfile=(\{a\}$, $\{a\}$, $\{a,d\}$, $\{a,b,c\}$, $\{a,b,c\}$,$\{a,b,c\}$, $\{a,b,c,d\}$, $\{a,b,c,d\})$ and let $k=3$. Adams-AV here picks only $\{a,\,b,\,c\}$. However, $\voters(c)$ is overrepresented as all its members have satisfaction $3>2=\lceil q_c\rceil$. If we switch $c$ with $d$, then not all voters in $\voters(d)$ would be overrepresented as the voter with ballot $\{a,\,d\}$ would have satisfaction $2=\lceil q_d\rceil$. Hence, Adams-AV and any refinement of it violate the symmetric variant of JUQ. By Theorem~\ref{thm:AUQ-thiele}, the same holds for every other composite Thiele rule, motivating our focus on JUQ. However, observe that the rule we introduce later on (Section~\ref{sec:UQER}), UQER, would satisfy the stronger symmetric variant of JUQ as a close inspection of the proof of \Cref{lem:UQER_JUQ} shows.

In light of \Cref{thm:AUQ-thiele} one might wonder whether Thiele rules, which maximize satisfaction, are just ill-suited to avoiding overrepresentation or whether JUQ is just hard to satisfy. The latter seems to be the case, as, to the best of our knowledge, no other known multi-winner voting rule satisfies JUQ. Indeed, in Appendix~\ref{app:counterexamples:JUQ}, we show counterexamples for a wide set of familiar rules (including Phragmén's rules and the method of equal shares, as well as all rules in the \texttt{abcvoting} Python package~\citep{joss-abcvoting} at the time of writing). Interestingly, even the sequential formulation of Adams-AV fails JUQ.

This is particularly important, as computing a winning committee under Adams-AV is $\NP$-hard because it refines the CC rule, which is \NP-hard to compute~\cite{PSZ08a}.\footnote{Indeed, hardness of CC is formulated by \citet{PSZ08a} as: ``Is there a committee $W$ with a CC-score of at least $s$?''. Since Adams-AV refines CC, a polynomial-time algorithm for Adams-AV would solve this question.} Thus, it is not clear whether we can find an exhaustive committee satisfying JUQ in polynomial time. In the next section, we try to tackle this question by investigating the dynamics of exchanging pairs of candidates that witness a violation of JUQ.

\subsection{Swap Dynamics of JUQ}\label{sec:JUQ-swap-dynamics}

We can think of JUQ as being defined in terms of \emph{swaps} of candidates, i.e., if there is a ``bad'' candidate~$c$, we should swap it with the ``good'' candidate~$d$ (whenever possible). Formally, a JUQ swap for a committee $W$ is a distinct pair of candidates $(c, d)\in\candidates^2$ such that $c\in W$ and $d\not\in W$ (together with some $N'\subseteq N(d)$) are witnesses to the violation of JUQ of $W$.

It is then natural to study the dynamics of such swaps to analyze the complexity of constructing an exhaustive committee that satisfies JUQ. Indeed, if we show a polynomial bound on the length of all possible sequences of JUQ swaps, then we have effectively shown the polynomial-time computability of JUQ. At the same time, studying the dynamics of JUQ swaps might give us further insights into the behavior of our axiom.

One first approach would be to show that, for any starting committee $W$, there is no sequence of JUQ swaps of form $(c_1, d_1), (c_2, d_2), \ldots, (c_{\ell-1}, d_{\ell-1}),(c_\ell, c_1)$, i.e., that it is never possible to add a candidate $c$ back in the committee after removing it through a JUQ swap. This would immediately give us a polynomial-time algorithm for computing an exhaustive JUQ committee: we could start from an arbitrary exhaustive committee $W$, perform at most $(|\candidates|-k)$-many swaps, and find a committee that satisfies JUQ. Unfortunately, this approach does not work in general.

\begin{example}\label{ex:can-swap-back}
	Consider the profile $\approfile=(\{a,b,c\},\{a,b,d\},\{e,f\},\{e,f\})$ and let $k=n$. If we start with the committee $\{a,b,c,d\}$ and swap first $a$ with $e$ and then $c$ with $f$, then $(d, a)$ will be a JUQ swap, and hence we can add $a$ back in.
\end{example}

In Appendix~\ref{app:swap-dynamics-extra}, we show that the above problem persists even assuming some clever schemata on how to pick which JUQ swap to perform next.

Another approach would be to use a bounded \emph{potential function}, and show that each swap decreases this function. As there are only finitely many committees, this would show that there are no infinite sequences of JUQ swaps. To obtain a polynomial-time algorithm one would then need to show that such sequences of swaps are polynomially-bounded. Given the scope of the paper, a natural starting point for such a potential would be to consider some measure of ``overrepresentation'' of a committee. To name a few natural candidates, one might consider the number of candidates $c$ for which $\voters(c)$ is overrepresented, the number of unique voters that are overrepresented, or the magnitude of total overrepresentation (i.e., sum of the satisfactions above quota).

Unfortunately, this approach also does not work---there are JUQ swaps that \emph{increase} overrepresentation. Consider the profile $\approfile = (\{a,b,c\}\times4,\,\{d,e,f\}\times3,\,\{d,g,h\}\times3,\,\{i\}\times6)$ and $k=8$.\footnote{The notation $\{a,b,c\}\times4$ means that there are four copies of the ballot $\{a,b,c\}$.} Hence, in this instance we have $n = 16$ and $\nicefrac{n}{k} = 2$. Consider the committee $W=\{a,b,c,e,f,g,h,i\}$. This committee violates JUQ, and $(c,d)$ is a JUQ swap. However, this change increases overrepresentation with respect to the natural measures discussed above. The number of candidates whose supporters are overrepresented increases: from $\{a,b,c\}$ to $\{e,f,g,h\}$. Similarly, the number of unique overrepresented voters grows from $4$ to $6$. The sum of the magnitudes of the violations also increases (from a total of $4$ points over upper quota to $6$).

Aside from the fact that this shows that another approach is required to show that the swaps eventually terminate, this might seem like a problem with JUQ: indeed, an axiom that aims to prevent overrepresentation should not require swaps that increase overrepresentation. However, we would argue that this is not as problematic as it might first seem. Indeed, from the example above, $W$ is the unique committee that minimizes the three measures we mentioned. On the other hand, this committee violates a rather weak notion of efficiency: it includes $e$ but not $d$, even though $\voters(e)\subset\voters(d)$. In other words, including $d$ instead of $e$ would make everyone equally satisfied or better off. So it appears that minimizing overrepresentation (at least in the sense above) is not a reasonable objective, as it clashes with an extremely weak and desirable notion of efficiency.

Still, this shows that we need a different potential function. The proof of Theorem~\ref{thm:JUQ-adams} already gives us a possible option: JUQ swaps always increase the Adams-AV-score of a committee (i.e., they increase the lexicographic objective function $(\score_{w_1}(\approfile, W), \score_{w_2}(\approfile, W))$). Since there are only finitely many committees, no infinite sequences can exist. 

\begin{restatable}{proposition}{JUQChains}
	There are no infinite sequences of JUQ swaps.\label{prop:JUQ-chain}
\end{restatable}

As a consequence, every sequence of JUQ swaps has length at most $\binom{m}{k}$, since visiting the same committee twice would lead to the possibility of an infinite sequence.
The fact that JUQ swap eventually terminate, however, is still not enough to show polynomial-time computability of JUQ-compliant exhaustive committees as, in principle, one could need super-polynomially many swaps to reach a (locally) optimal committee where no JUQ swaps are possible.\footnote{For example, for PAV, one might need super-polynomially many swaps to converge~\cite{KrEl24a}.} We leave a precise characterization of such dynamics for future work. In the next section, we tackle the complexity of constructing JUQ committees from another angle.

\subsection{The Upper Quota Elimination Rule}\label{sec:UQER}

In addition to Thiele rules such as PAV, several sequential rules that aim to prevent underrepresentation have been introduced in the multi-winner voting literature.
In particular, we recall the definition of the Greedy Justified Candidate Rule (GJCR) \citep{BrPe23a}, which satisfies EJR+ (but fails JUQ, see Appendix~\ref{app:counterexamples:JUQ}).

\begin{definition}\label{def:GJCR}
    GJCR begins with the empty committee $W=\emptyset$ and iteratively adds candidates witnessing an EJR+ violation to it. Among all EJR+ violations witnessed by a candidate $d$ and a non-empty set of voters $N'\subseteq N(d)$, GJCR picks the one maximizing $\lvert N'\rvert$ and adds $d$ to $W$. The rule terminates once no EJR+ violations remain, returning a committee $W$ with $|W|\leq k$.
\end{definition}

We take a similar sequential approach to devise the \emph{Upper Quota Elimination Rule} (UQER), showing that exhaustive committees satisfying JUQ can be found in polynomial time in the process. We take the ``dual approach'' to GJCR: instead of adding EJR+ witnesses, our rule starts off with the entire set of candidates $C$ and then iteratively deletes candidates witnessing an upper quota violation. This is repeated until either no candidate is overrepresented (and hence the remaining set of candidates satisfies UQ) or until the rule reaches a committee of size $k$. We will show that this committee and any set considered in between indeed satisfies JUQ.

Formally, our rule proceeds as follows. We start with $W_0=\candidates$, the set of all candidates. In each step $j=1,2,\dots$ we choose a candidate to remove. To do so, we find all candidates $c\in W_{j-1}$ that cause a UQ violation for $W_{j-1}$, that is $\lvert A_i \cap W_{j-1} \rvert > \lceil q_c\rceil$ for all $i \in N(c)$. Let $c_j$ be the candidate among these with the fewest supporters $\lvert N(c)\rvert$, with ties broken arbitrarily. Set $W_{j}=W_{j-1}\setminus \{c_j\}$ and repeat. We continue until the first step $j^*$ in which either $|W_{j^*}|=k$ or $W_{j^*}$ satisfies UQ. In the first case, our rule outputs $W=W_{j^*}$, and in the second case it outputs an arbitrary subset $W\subseteq W_{j^*}$ with $|W|=k$. We could, for instance, remove the $|W_{j^*}|-k$ candidates with fewest approvals.\footnote{We note that UQER is similar to the dual-quota apportionment method of \citet{Mayb78a}. However, even for apportionment, UQER differs as we explicitly use the committee size $k$, and restrict our deletions to candidates who are strictly (not weakly) above the upper quota. Further, \citet{Mayb78a} uses a different method to select which candidates to delete.}

\begin{theorem}\label{thm:UQER}
    UQER returns an exhaustive committee satisfying JUQ in polynomial time.
\end{theorem}

To prove this claim we show that JUQ is always maintained during the execution of UQER.

\begin{lemma}\label{lem:UQER_JUQ}
    Every intermediate set of candidates in $\{W_0,W_1,\dots, W_{j^*}\}$ produced during the execution of UQER satisfies JUQ.
\end{lemma}

\begin{proof}
    It is easy to see that the initial set $W_0=C$ satisfies JUQ, as there are no candidates outside the set that can be part of a JUQ violation. We can also observe that the set of UQ-violating candidates never increases during the execution of UQER.
	
    Suppose at some point we encounter a JUQ violation in the set $W_j$, witnessed by some candidate pair $(c,d)$ and voter set $N'\subseteq \voters(d)$.\footnote{Note that $c$ need not be the same as $c_{j+1}$, i.e., the candidate that will be removed from $W_j$ to obtain $W_{j+1}$.} This means that $|A_i\cap (W_j\setminus \{c\}\cup \{d\})|\leq \ceil{q(N')}\leq \ceil{q_{d}}$ for all $i\in N'$, and thus $|A_i\cap (W_j\setminus \{c\})|< \ceil{q_{d}}$.

    As $d \notin W_{j}$ it must have been removed in some earlier step $x$, i.e. $d=c_x$ for some $x<j$. As $d$ was removed in this step, we know that $|A_i\cap W_{x-1}|> \ceil{q_{d}}$ for all $i\in \voters(d)$ as $d$ was a UQ-violating candidate for $W_{x-1}$. Thus, after $d$ was removed, $|A_i\cap W_{x}|\geq \ceil{q_{d}}$ for all $i\in \voters(d)$. Further, by our selection of $d$ we know that $d$ was the overrepresented candidate with the least approvals. Hence, $q_d \le q_{d'}$ for any other candidate $d'$ whose supporters are overrepresented at step $x$. Now let $i \in N' \setminus N(c)$ be any voter from the underrepresented group not approving $c$. Note that by the same argument as in \Cref{thm:JUQ-adams} such a voter must exist. As $\lvert A_i \cap W_x\rvert \ge \lceil q_d\rceil$ and $\lvert A_i \cap W_{j}\rvert = \lvert A_i \cap W_{j}\setminus \{c\}\rvert < \lceil q_d\rceil$, there must exist a candidate $c' \in A_i \cap (W_x \setminus W_{j})$. Let $c'$ be the last such removed candidate and assume it was removed at a step $y$ between step $x$ and step $j$. As it was the last, we know that $\lvert A_i \cap (W_{y-1}\setminus \{c'\})\rvert < \lceil q_d\rceil$.  By our assumption that $d$ was a candidate with the least approvals when removed at step $x$ we know that $q_{c'} \ge q_d$.  In particular, for this step it must hold that $\lvert A_i \cap W_{y-1} \rvert > \lceil q_{c'} \rceil \ge \lceil q_d\rceil $, since the set of overrepresented candidates never increases. This, however is a contradiction to $\lvert A_i \cap (W_{y-1}\setminus \{c'\})\rvert < \lceil q_d\rceil$ and thus $(c,d)$ could not have been a JUQ violation.
	Thus every candidate set from $\{W_0,W_1,\dots, W_{j^*}\}$ satisfies JUQ.
\end{proof}

With this, we can easily prove our theorem.

\begin{proof}[Proof of \Cref{thm:UQER}]
    Using \Cref{lem:UQER_JUQ} we know that $W_{j^*}$ satisfies JUQ. If $|W_{j^*}|=k$ then we are done. Meanwhile, if $|W_{j^*}|>k$ we know that $W_{j^*}$ additionally satisfies UQ. Thus, any subset $W\subseteq W_{j^*}$ will satisfy UQ and thus also the weaker JUQ. Finally, UQER runs in polynomial time, as we only need to check which candidates are overrepresented for at most $m-k$ iterations.
 \end{proof}

Similarly to \Cref{lem:UQER_JUQ}, we can show UQER maintains EJR+ during its execution, until the final step. We will use this result later, in \Cref{sec:combining}.  Note that the result does not necessarily hold for the arbitrary subset $W \subseteq W_{j^*}$ that is returned if $|W_{j^*}| > k$.

\begin{restatable}{lemma}{UQEREJR}\label{lem:UQER_EJR}
    Every intermediate set of candidates $\{W_0,W_1,\dots, W_{j^*}\}$ produced during the execution of UQER satisfies EJR+.
\end{restatable}

\section{Balancing Under- and Overrepresentation}\label{sec:JNQ}

We now turn to our second view on overrepresentation: near quota.  Intuitively, there should not exist a party $P_i$ which is overrepresented and a party $P_j$ which is underrepresented such that by moving a seat from party $P_i$ to party $P_j$ we can bring both parties closer to their respective quotas. Our goal is to generalize this to the multi-winner voting setting, by again interpreting single candidates as parties that represent their voters. In this spirit, there should not exist a candidate inside the committee whose approvers are overrepresented and a candidate outside the committee for whom a subset of the approvers are underrepresented such that we could swap these two candidates and make both ``better represented''. To better interpret this, we again return to the apportionment setting. We first give the following equivalent formulation of apportionment near quota: two parties $P_i$ and $P_j$ witnessing a near quota violation is equivalent to the first party being more than $0.5$ above its quota, and the second party being more than $0.5$ below its quota.
\begin{restatable}{observation}{nearquotaequiv}\label{obs:nq}
	An outcome $W$ in an apportionment instance satisfies apportionment near quota if and only if there are no two parties $P_i$ and $P_j$ such that 
	\(
	 a_i(W) - q(P_i) > \frac{1}{2} \text{ and } q(P_j) - a_j(W) > \frac{1}{2}.
	\)
\end{restatable} 
To generalize this to the approval setting, we have to account for potentially overlapping approval sets.  Assume that $c \in W$ is the overrepresented candidate and $d \in C\setminus W$ is the underrepresented one and let $N'\subseteq N(d)$ be the set of underrepresented voters approving $d$. Then we say that $c$ and $d$ witness a \emph{justified near quota} violation if swapping $c$ with $d$ moves every voter in $N(c)$ closer to their quota $q_c$ \emph{whenever their satisfaction changes} (i.e., when they do not also approve of $d$) and every voter in $N'$ moves closer to their quota $q(N')$ whenever their satisfaction changes. For voters approving of both of $c$ and $d$ the satisfaction does not change. For these voters we merely require that they are above the quota $q_c$ and below the quota $q(N')$ to ensure these voters are not on the ``wrong side'' of their respective quotas.

In the following, $\indicator[\cdot]$ denotes the indicator function, i.e., $\indicator[p]=1$ if $p$ is true and $0$ otherwise.
\begin{definition}\label{def:JNQ}
	A set $W$ satisfies \emph{justified near quota} (JNQ) if there does not exist a pair of candidates $c\in W$ and $d\in C\setminus W$, and non-empty set of voters $\voters^\prime\subseteq \voters(d)$ such that
	\begin{align*}
		 \lvert A_i \cap W\rvert - q_c&> \frac{1}{2} \cdot \indicator[d \notin A_i]\quad\quad\text{ for all }i\in \voters(c),\text{ and }\\
		 q(\voters^\prime) -  \lvert A_i \cap W\rvert&> \frac{1}{2} \cdot \indicator[c \notin A_i]\quad\quad\text{ for all }i\in\voters^\prime.\qedhere
	\end{align*}
\end{definition}

The following example shows that, as one would expect, JNQ and JUQ are two distinct notions.

\begin{example}
	Consider the profile $\approfile=(\{a\},\{a\},\{a\},\{a,c\}, \{b\})$ and let $k=2$, so the quota of a single voter is $\nicefrac{2}{5}$. The committee $\{a, c\}$ satisfies JNQ. Indeed, replacing either $a$ or $c$ with $b$ would change the satisfaction of the voter submitting $\{b\}$ from $0$ to $1$, increasing the distance from their quota of $\nicefrac{2}{5}$. On the other hand, this committee fails JUQ as the voter approving $\{a,c\}$ is overrepresented and we could swap $c$ with $b$. Conversely, the committee $\{b,c\}$ satisfies JUQ, but fails JNQ. Indeed, here we should swap $b$ with $a$, and the set $\voters^\prime\subseteq\voters(a)$ can consist for example of voters~$1$ and~$2$. Finally, the committee $\{a,b\}$ satisfies both axioms.
\end{example}

Indeed, following \Cref{obs:nq} we can see that JNQ generalizes apportionment near quota.

\begin{restatable}{proposition}{JNQApportionment}
	For apportionment instances, JNQ is equivalent to apportionment near quota.\label{prop:jnq-anq}
\end{restatable}

Next, similarly to Adams-AV and JUQ, we get that every refinement of SLAV satisfies JNQ. The proof is analogous to the Adams-AV proof. That is, we show that if a committee $W$ does not satisfy JNQ as witnessed by two candidates $c$ and $d$ then swapping $c$ with $d$ increases the SLAV score.

\begin{restatable}{theorem}{slav}\label{thm:SLAV_JNQ}
	Any refinement of SLAV satisfies justified near quota.
\end{restatable}

Using this, we again obtain a characterization, this time of refinements of SLAV. 
\begin{theorem}
	A composite Thiele rule satisfies apportionment near quota if and only if it is a refinement of SLAV.
\end{theorem}
\begin{proof}
	The first direction of the theorem follows from \Cref{thm:SLAV_JNQ} and the fact that JNQ and apportionment near quota are equivalent on apportionment instances (\Cref{prop:jnq-anq}).
	
	We show the second direction by proving that for any simple Thiele rule except for SLAV there exists an apportionment instance on which it selects only a single outcome which violates apportionment near quota. As it only selects a single outcome, every refinement of this rule must also select this single outcome and therefore violate apportionment near quota.
	
	Consider any other Thiele vector $w = (1, w(2), w(3), \dots)$ and let $w(t)$ be the first coordinate in which $w$ differs from SLAV. We consider two cases, either $w(t) > \frac{1}{2t-1}$ or $w(t) < \frac{1}{2t-1}$.
	
	\emph{Case 1: } $w(t) > \frac{1}{2t-1}$ or equivalently $\frac{1}{w(t)} < 2t-1$. We let $\frac{p}{q} \in \rationals$ such that $\frac{1}{w(t)} < \frac{p}{q} < 2t-1$ and consider an apportionment instance with two parties and votes $(p,q)$ and $k = t$. We note that $p \cdot w(t) > q$. Hence, the Thiele rule would give all $t$ seats to the first party. However, the quota of the first party is $t \cdot \frac{p}{p + q} = t \cdot \frac{1}{1 + \frac{q}{p}} <  t \cdot \frac{1}{1 + \frac{1}{(2t-1)}} = t - \frac{1}{2},$ while the quota of the second party is  $t \cdot \frac{q}{p + q} = t \cdot \frac{1}{\frac{p}{q} + 1} >t \cdot \frac{1}{2t - 1 + 1} = \frac{1}{2}.$
	Hence, as a consequence, these two parties violate near quota.
	
	\emph{Case 2: } $w(t) < \frac{1}{2t-1}$ or, equivalently, $\frac{1}{w(t)} > 2t-1$. We again let $\frac{p}{q} \in \rationals$ such that $\frac{1}{w(t)} > \frac{p}{q} > 2t-1$ and again consider an apportionment instance with two parties and votes $(p,q)$ and committee size $t$. For this, we claim that the seat distribution is $(t - 1, 1)$. Indeed, it holds that $q > w(t) \cdot p$ and one can verify that (since $t$ is the first instance for which the vector differs from SLAV) $\frac{q}{3}  < \frac{p}{2t-3}$ and thus $(t-2, 2)$ would not be a valid outcome. 
	
	For the quotas we compute
		$t \cdot \frac{p}{p + q} = t \cdot \frac{1}{1 + \frac{q}{p}} > t \cdot \frac{1}{1 + \frac{1}{(2t-1)}} = t - \frac{1}{2}$ and $t \cdot \frac{q}{p + q} = t \cdot \frac{1}{\frac{p}{q} + 1} < t \cdot \frac{1}{2t - 1 + 1} = \frac{1}{2}.$
	Thus, this Thiele method also violates near quota. As a consequence, any composite Thiele rule satisfying apportionment near quota must be a refinement of SLAV.
\end{proof}

Going beyond Thiele rules, in Appendix~\ref{app:counterexamples:JNQ} we detail counterexamples that show that essentially all familiar voting rules fail JNQ (e.g., the Method of Equal Shares, Phragmén's rule, and all the rules in the \texttt{abcvoting}~\citep{joss-abcvoting} Python package at the time of writing). In particular, the variance-Phragmén voting rule (which is equivalent to Sainte-Laguë's method on apportionment instances) also fails JNQ. 
As an immediate consequence of Theorem~\ref{thm:SLAV_JNQ}, analogously to JUQ swaps, we get that there are no infinite sequences of JNQ swaps, as any swap increases the SLAV-score.

\begin{restatable}{corollary}{JNQChain}
	There are no infinite sequences of JNQ swaps.\label{cor:JNQ-chain}
\end{restatable}

Finally, similarly to JUQ and EJR+, we show that it can be verified in polynomial-time whether a given committee satisfies JNQ.

\begin{restatable}{proposition}{JNQVerifPoly}
   Whether a committee satisfies JNQ is decidable in polynomial time.
\end{restatable}

It is open whether exhaustive committees satisfying JNQ can be computed in polynomial time.

\section{Combining Lower, Near, and Upper Quota Axioms}\label{sec:combining}

In this section, we investigate whether our upper, near, and lower quota axioms can be simultaneously satisfied in the multi-winner voting setting.

First, we observe that, as a refinement of CC, Adams-AV satisfies the rather weak underrepresentation axiom of \emph{justified representation} \citep{ABC+16a}, in addition to satisfying JUQ and being exhaustive. Can we do better? In particular, is it always possible to find an (exhaustive) rule that satisfies EJR+ and JUQ, or maybe even EJR+, JUQ and JNQ?

Unfortunately, in the class of Thiele rules, this is not the case. We have seen a ``tripartition'' with regard to representation axioms. PAV, and its refinements, are the only composite Thiele rules satisfying EJR+, SLAV refinements are the only ones satisfying JNQ, and Adams-AV refinements are the only ones satisfying JUQ, with these (im)possibilities mirroring the apportionment setting.

For the apportionment setting, however, \citet[Proposition~6.5]{BaYo01a} prove a curious fact about divisor methods: any committee selected by a divisor method either satisfies lower quota or upper quota. Interestingly enough, we can generalize this result to the approval-based multi-winner voting setting and show that for a large class of Thiele rules (corresponding to the divisor methods as defined by \citeauthor{BaYo01a}) including PAV and SLAV, every selected committee satisfies either UQ or EJR+.

\begin{theorem}
	Let $w = (w(1), w(2), \dots)$ be a weight function with strictly positive weights for which there exists an $\alpha \in \posReals$ satisfying $(i - 1) < \frac{\alpha}{w(i)} \le i$ for all $i \in \posNats$ or $(i - 1) \le \frac{\alpha}{w(i)} < i$ for all $i \in \posNats$. Then any committee selected by the simple $w$-Thiele rule satisfies either UQ or EJR+.
\end{theorem}
\begin{proof}
	Assume that $W$ is selected by $w$-Thiele and that it fails both UQ and EJR+. Let $c$ and $d$ be the respective witnesses of this.
	Thus, it holds that $\lvert A_i \cap W\rvert > \lceil q_c\rceil$ for every $i \in \voters(c)$ as well as $\lvert A_i \cap W\rvert < \lfloor q(N')\rfloor$ for some non-empty $N' \subseteq \voters(d)$ and every $i \in N'$.
	
	Now consider the score change by exchanging $c$ with $d$. The score of every voter $ i \in N' \setminus N(c)$ increases by $w(\lvert A_i \cap W\rvert + 1)$, while the score of every $ i \in N(c) \setminus N'$ decreases by  $w(\lvert A_i \cap W\rvert)$. With this we obtain
	\begin{align*}
		\sum_{i \in N' \setminus \voters(c)} w(\lvert A_i \cap W\rvert + 1) - \sum_{i \in \voters(c) \setminus N'} w(\lvert A_i \cap W\rvert)
		&\ge 	\sum_{i \in N'} w(\lvert A_i \cap W\rvert + 1) - \sum_{i \in \voters(c)} w(\lvert A_i \cap W\rvert) \\ 
		\ge \sum_{i \in N'} w(\lfloor q(N') \rfloor) - \sum_{i \in \voters(c)} w(\lceil q_c\rceil + 1)
		& > \lvert N'\rvert \frac{\alpha}{q(N')} - n_c \frac{\alpha}{q_c} = \alpha(\frac{n}{k} - \frac{n}{k}) = 0.
	\end{align*}
	Hence, by swapping $c$ with $d$ the score would increase and thus $W$ was not an optimal committee according to $w$-Thiele.
\end{proof}
As a corollary, we immediately obtain that this is true for PAV, SLAV, but also Adams-AV.
\begin{restatable}{corollary}{corthielerules}
	Any committee selected by PAV, SLAV, or Adams-AV satisfies either UQ or EJR+. 
\end{restatable}
\begin{proof}
	For PAV we observe that it satisfies the definition with $\alpha = 1$ and SLAV with $\alpha = \frac{1}{2}$. For Adams-AV we observe that, for a fixed instance with $n$ voters and target committee size $k$, Adams-AV is equivalent to the simple Thiele rule with vector $(n\cdot k, 1, \frac{1}{2}, \dots)$ (see Appendix~\ref{app:adams:thiele}). As this Thiele rule satisfies the condition with $\alpha = 1$, this implies the result for Adams-AV. 
\end{proof}
As another consequence, these rules satisfy EJR+ on instances in which no UQ committee exists.
The aforementioned classical result of \citet{BaYo01a} was partially generalized for a restricted class of divisor methods by \citet{cembrano2025new}, who showed that for any fixed apportionment instance there exists a divisor method satisfying both upper and lower quota in that instance. An interesting open question is whether this also generalizes to multi-winner voting: for any instance does there exist a Thiele method that satisfies EJR+ and JUQ in that instance?

If we consider arbitrary multi-winner voting rules, we can show that UQ (and thus JUQ), JNQ and EJR+ are jointly satisfiable, as long as we do not require exhaustiveness. To this end, we define a voting rule, the \emph{Proportional Elimination Rule}, in a similar way to UQER, but with a different initial committee and a modified method of selecting candidates for removal.

We start from any committee $W_0$ output by GJCR (\Cref{def:GJCR}). This committee has size $|W_0|\leq k$. In each step $j\in \{1,2,\dots\}$, we find all candidates $c\in W_{j-1}$ for which all supporters $i\in N(c)$ are strictly above their lower quota, i.e. $|A_i\cap W_{j-1}|>\floor{q_c}$.\footnote{Note that this is different to UQER, which removed candidates violating UQ, i.e. those with all supporters above their upper quota.} 
If there are no such candidates then we are done and our voting rule outputs $W^*=W_{j-1}$. Otherwise, let $c_j$ be the candidate from among these with the least approvals (with ties broken arbitrarily), set $W_{j}=W_{j-1}\setminus \{c_j\}$, and repeat.

\begin{proposition}
    The (possibly not exhaustive) committee $W$ produced by the Proportional Elimination Rule satisfies UQ, JNQ and EJR+. 
\end{proposition}

\begin{proof}
    
    Our rule ensures that for each $c\in W$ there exists $i\in N(c)$ with $|A_i\cap W|\leq \floor{q_c}\leq \ceil{q_c}$, and thus satisfies UQ by construction.

    Suppose that $W$ violates JNQ, witnessed by $c\in W$, $d\in C\setminus W$ and non-empty $N'\subseteq \voters(d)$. Thus, for all $i\in N(c)$, we have $|A_i\cap W|>q_c$. However, the rule ensures $|A_i\cap W|\leq\floor{q_c}$ for some $i \in N(c)$ and thus $|A_i\cap W|\leq q_c$, hence $c$ cannot be part of a JNQ violation.
     
    Suppose now that $W$ violates EJR+, as witnessed by $d\in C\setminus W$ and $N'\subseteq \voters(d)$. Thus, we know that $\lvert A_i \cap W\rvert < \lfloor q(N') \rfloor$ for every voter $i\in \voters'$. Note that the committee $W_0$ output by GJCR satisfies EJR+. We let $c^*\in W_0$ be the first candidate after whose deletion the EJR+ violation with $d$ occurs. Call the committee after $c^*$ was deleted $W'$. As $W'$ is the first committee witnessing the EJR+ violation, there must exist a voter $i \in N'$ such that $\lvert A_i \cap W' \rvert < \lfloor q(N') \rfloor$ and $\lvert A_i \cap (W' \cup \{c^*\}) \rvert \ge \lfloor q(N') \rfloor$ as otherwise the violation would have already occurred before. This, in particular, implies that $i \in N(c^*)$. 
    Hence, by definition of the Proportional Elimination Rule for this same voter we get that  $\lfloor q_{c^*}\rfloor \leq \lvert A_i \cap W'\rvert < \lfloor q(N') \rfloor \leq \lfloor q_d \rfloor$. Note that this implies that $\lvert N(c^*)\rvert <\lvert N'\rvert \leq \lvert N(d) \rvert$.
    
    First, suppose that that $d\in W_0$. As $d \notin W'$, we know that $d$ must have been removed before $c^*$. However, at the step in which $d$ was removed, $c^*$ would also have been eligible to be removed. As $\lvert N(c^*)\rvert < \lvert N(d)\rvert $ our rule would have removed $c^*$ over $d$, leading to a contradiction. 
    
    Now, instead assume that $d\notin W_0$, that is $d$ was not selected by GJCR. Let $W''$ be the set of candidates selected by GJCR before choosing the first candidate with weakly fewer than $\lvert N(c^*)\rvert$ supporters witnessing its EJR+ violation. To obtain $W'$ from $W_0$, the Proportional Elimination Rule only removed candidates with weakly fewer than $\lvert N(c^*)\rvert$ supporters, which means $W''\subseteq W'$. Therefore, $\lvert A_i \cap W''\rvert \leq \lvert A_i \cap W'\rvert < \lfloor q(N') \rfloor$ for all voters $i\in N'$, and hence, $d$ and $N'$ would have been an EJR+ violation for $W''$. This is a contradiction to our construction that the next candidate chosen by GJCR has at most $\lvert N(c^*)\rvert< \lvert N'\rvert$ supporters witnessing the violation.
\end{proof}

With exhaustiveness, the picture gets more complicated. 
Indeed, it remains an open problem whether any of the three pairs (JUQ, JNQ), (JUQ, EJR+) and (JNQ, EJR+) are jointly satisfiable for exhaustive committees.
In the following, we explore the interplay of JUQ, EJR+ and exhaustiveness in more detail.
For other notions of proportionality, however, we can give a clear answer. We show in Appendix~\ref{app:additional-more-lower-quota} that JUQ is incompatible with several other popular axioms for multi-winner voting, namely priceability, fully justified representation, and perfect representation. The latter two impossibilities hold even without requiring exhaustiveness.

Our first result shows that JUQ, EJR+ and exhaustiveness are ``almost'' compatible, i.e., we can either find an exhaustive committee satisfying EJR+ and JUQ or we can always find two sets of candidates satisfying both EJR+ and UQ, one with $\leq k$ candidates and one with $\geq k$ many.

\begin{corollary}\label{cor:ejr_strong_uq}
	There always exists a committee $W_\leq$ satisfying EJR+ and UQ with $|W_\leq |\leq k$. Additionally, one of the following is always true:
    \begin{itemize}
        \item There exists a committee $W_=$ satisfying EJR+ and JUQ with size exactly $|W_=|=k$.
        \item There exists a set of candidates $W_\geq$ satisfying EJR+ and UQ with size $|W_\geq|\geq k$. 
    \end{itemize}
\end{corollary}

\begin{proof}
    $W_\leq$ is output by the Proportional Elimination Rule.
To find either $W_=$ or $W_\geq$, recall from \Cref{thm:UQER} and \Cref{lem:UQER_EJR} that during the execution of UQER we find a candidate set $W_{j^*}$ that satisfies EJR+ and JUQ and, additionally, either (1) has $|W_{j^*}|=k$, giving us $W_==W_{j^*}$, or (2) satisfies UQ and has $|W_{j^*}|>k$, giving us $W_\geq=W_{j^*}$.
\end{proof}

A voting rule satisfying JUQ, EJR+ and exhaustiveness is unlikely to be as straightforward as some of the voting rules we proposed earlier. One might be tempted to modify the outcome of a rule that satisfies two of EJR+, JUQ, and exhaustiveness to also satisfy the third. Below we provide negative results for three of the most natural modifications. First, we show that the Proportional Elimination Rule cannot be made exhaustive (while satisfying JUQ) by adding additional candidates.

\begin{wraptable}[13]{R}{0.4\textwidth}
	\centering
	\small
	\setlength{\tabcolsep}{3pt}
	\begin{tabular}{ccccccccccccc}
			\toprule
		Candidate & $3$ & $3$ & $3$ & $14$ & $7$ & $7$ & $14$ & $12$ \\
		\midrule
		$c_1,c_2$& & & & \cmark & \cmark & & & \cmark  \\
		$c_3, c_4$& & & & & & \cmark & \cmark \\
		$c_5$& \cmark & \cmark & \cmark \\
		$c_6$&  \cmark & \cmark & & \cmark \\
		$c_7$&  \cmark & & \cmark & & \cmark & \cmark \\
		$c_8$&  & \cmark & \cmark & & & & \cmark  \\
			\bottomrule
	\end{tabular}
	\captionsetup{justification=centering}
	\caption{Example instance for \Cref{ex:PER_Completion} with $k=7$. The column headers refer to the number of voters submitting that ballot (e.g., there are $3$ voters approving $\{c_5,c_6,c_7\}$).}
	\label{tab:PER_Completion}
\end{wraptable}

\begin{proposition}\label{ex:PER_Completion}
    There exists an instance where every exhaustive superset of the outcome of the Proportional Elimination Rule fails JUQ.
\end{proposition}
\begin{proof}
    Consider the instance with $n=63$ and $k=7$ depicted in Table~\ref{tab:PER_Completion}.
    Here, GJCR selects only the non-exhaustive committee $W=\{c_1,c_2,c_3,c_4,c_5\}$. Notice that, for every $c\in W$ and $i\in\voters(c)$, we have that $|A_i\cap W|\leq\floor{q_c}$. Thus the Proportional Elimination Rule also selects $W$.

Now consider any possible completion of this committee to 7 candidates. Without loss of generality, let us look at $W^\prime=\{c_1,\dots, c_7\}$. Then $c_5$ and $c_8$ witness a JUQ violation with $\voters^\prime=\voters(c_8)$. Indeed, $q_{c_5}=1$ and every voter $i\in N(c_5)$ has $|A_i\cap W^\prime|=2$. Moreover, $\ceil{q_{c_8}}=3$ and every voter $i\in N(c_8)$ has $|A_i\cap (W^\prime\setminus\{c_5\}\cup \{c_8\})|\leq 3$.
\end{proof}

In other words, we cannot always start with $W_\leq$ from \Cref{cor:ejr_strong_uq} and add $k-|W_\leq|$ candidates to this committee while continuing to satisfy EJR+. 
Analogously, we cannot always remove $|W_\geq|-k$ candidates from $W_\geq$ from \Cref{cor:ejr_strong_uq} while continuing to satisfy JUQ.

\begin{restatable}{proposition}{uqerelimination}\label{prop:UQER_elimination}
    The partial outcome $W_{j^*}$ of UQER (\Cref{sec:UQER}) need not contain a size $k$ subset that satisfies EJR+.
\end{restatable}

Finally, we consider whether we can start from an exhaustive committee satisfying EJR+, and modify it by \emph{cleverly} choosing JUQ swaps, such that EJR+ is maintained. Unfortunately, this is not always possible. Consider a multi-winner election with $k=5$ and $\approfile=(\{a,b\}\times 3,\,\{a,b,c\},\,\{c,d,e\},\,\{d,e\}\times 3,\,\{f\},\,\{g\})$. Committee $W=\{a,c,d,f,g\}$ satisfies EJR+. However $c$ witnesses a JUQ violation with either $b$ and voter set $\voters(b)$ or $e$ and voter set $\voters(e)$. These are the only JUQ violations, and satisfying either would lead to a committee that fails EJR+. This claim holds even if we select a particularly well chosen EJR+ committee, such as the output of the Method of Equal Shares. In the interest of space, we refer the definition of this rule to Appendix~\ref{app:omitted_proofs}.

\begin{restatable}{proposition}{mesjuqswaps}\label{prop:MES_JUQ_SWAPS}
    A sequence of JUQ swaps that starts from a committee satisfying EJR+ might terminate with a committee that violates EJR+. This holds even if the initial committee was selected by the Method of Equal Shares (with any completion method), and even if we use the profile of a candidate's supporter satisfactions\footnote{For some candidate $c\in W$ this profile would consist of the multiset $\{|A_i\cap W|\}_{i\in N(c)}$.} to select which JUQ swap to perform next.
\end{restatable}

\section{Conclusion and Research Directions}\label{sec:conclusion}

In this paper, we studied overrepresentation in approval-based multi-winner voting. We defined an axiom aimed at avoiding overrepresentation whenever possible, called justified upper quota (JUQ). We further introduced the class of composite Thiele rules and showed that, among this class, Adams-AV and its refinements are the only rules satisfying JUQ. Moreover, we showed that JUQ is satisfiable in polynomial time via the Upper Quota Elimination Rule. Additionally, we introduced an axiom that aims at balancing under- and overrepresentation, justified near quota (JNQ), and showed that, among the class of composite Thiele rules, only SLAV and its refinements satisfy it. Finally, we discussed the relationship between our notions and established lower quota notions, such as EJR+. We note that some of our results, such as the (partial) characterization of Adams-AV and SLAV through JUQ and JNQ (respectively), mirror and extend those obtained in the classical apportionment literature \citep{BaYo01a}. 

We identify several interesting open questions for future work. A first task would be to characterize the relationship between JUQ and JNQ for exhaustive committees and, similarly, the relationship between our axioms and EJR+. In other words, does there exist an exhaustive rule that avoids under- and overrepresentation at the same time? We know from our results that this is not possible within the composite Thiele class and that such a rule must fail the commonly studied axiom of priceability. Next, we note that all other established voting rules we have examined fail both our axioms. This suggests that our axioms embody a rather different approach to multi-winner voting rules compared to the existing literature. Therefore, we believe that finding natural and normatively-appealing rules that satisfy our axioms could lead to fundamentally new insights in multi-winner voting. Finally, we point out that approval voting is only one possible ballot type in multi-winner voting. It would be interesting to explore overrepresentation for ordinal or cardinal preferences or even for subset selection problems beyond multi-winner voting, such as recently done for clustering by \citet{jerrett2025low}.


\begin{acks}
	\includegraphics[height=\fontcharht\font`\B]{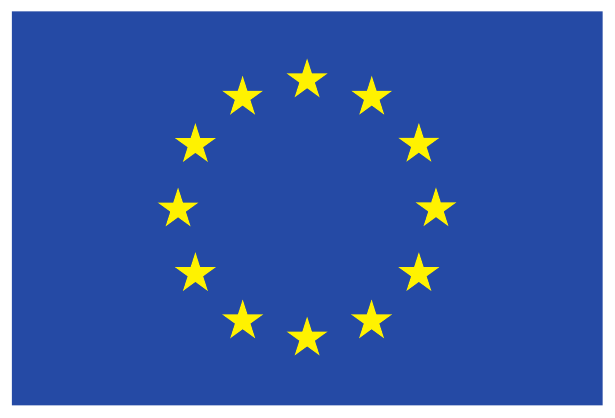} Oliviero Nardi has received funding from the \grantsponsor{101034440}{European Union's Horizon 2020 research and innovation programme}{} under grant agreement No~\grantnum{101034440}{101034440}.
This research was funded in whole or in part by the \grantsponsor{10.55776}{Austrian Science Fund (FWF)}{} (grants \grantnum{10.55776/PAT7221724}{10.55776/PAT7221724} and \grantnum{10.55776/COE12}{10.55776/COE12}), by netidee Förderungen, Austria (\url{https://www.netidee.at}) and by the \grantsponsor{10.47379}{Vienna Science and Technology Fund (WWTF)}{} (grant \grantnum{10.47379/ICT23025}{10.47379/ICT23025}). Jannik Peters received funding by the Singapore Ministry of Education under grant number MOE-T2EP20221-0001.
We thank Markus Brill for helpful discussions.
\end{acks}


\bibliographystyle{ACM-Reference-Format}
\bibliography{abb,algo,references}

\newpage
\appendix

\section{Omitted Proofs and Definitions} \label{app:omitted_proofs}
\obsejrpeq*
\begin{proof}
	Let $(\approfile,k)$ be an apportionment instance with parties $P_1,\dots,P_t$, and let $W$ be a committee.
	First, assume that $W$ violates apportionment lower quota. Then there exists a party $P_\ell$ such that
	$a_\ell(W) < \lfloor q(P_\ell)\rfloor$.
	Since $|A_i|=k$ for $i\in P_\ell$ and $a_\ell<k$, there exists a candidate $c\in A_i\setminus W$.
	Since $P_\ell = \voters(c)$ and $|A_i\cap W| = a_\ell < \lfloor q_c\rfloor$, this also witnesses a violation of EJR+.
	
	On the other hand, assume that $W$ violates EJR+. Then there exist $c\in C\setminus W$ and a non-empty set $N'\subseteq N(c)$ such that $|A_i\cap W| < \lfloor q(N')\rfloor$ for all $i\in N'$.
	Let $P_\ell$ be the party of voters approving $c$. Then $N(c)=P_\ell$, and by definition of apportionment instances $a_\ell < \lfloor q(N')\rfloor \le \lfloor q(P_\ell)\rfloor$. Hence, $W$ violates apportionment lower quota for $P_\ell$.
\end{proof}

\adamsextends*

\begin{proof}
	In Appendix~\ref{app:adams:thiele} we observe that, for a fixed instance with $n$ voters and target committee size $k$, Adams-AV is equivalent to the simple Thiele rule with weight vector $w=(n\cdot k, 1, \nicefrac{1}{2}, \nicefrac{1}{3}, \dots)$. Hence, by Theorem~1 of \citet{BLS18a}, we know that Adams-AV extends the divisor method $F$ defined by $d=(\nicefrac{1}{nk}, 1, 2, 3, \ldots)$. It remains to be shown that $F$ is equivalent to Adams' method.

	Recall that Adams' method starts by assigning a seat to every party (in order of party size) and continues until all seats are assigned or until all parties have received at least one seat. Then, during the execution of $F$, let $P_i$ be the largest party such that $a_i=0$. For every party $P_j$ with $a_j=0$,
	\[
		\nicefrac{|P_i|}{d(0)}=n\cdot k\cdot|P_i|\geq n\cdot k\cdot|P_j|=\nicefrac{|P_j|}{d(0)}.
	\]
	On the other hand, for every party $P_j$ with $a_j>0$,
	\[
		\nicefrac{|P_i|}{d(0)}=n\cdot k\cdot|P_i|>n-1\geq |P_j|\geq |P_j|\cdot\nicefrac{1}{d(a_j)}.
	\]
	Thus, as long as there are parties that have received no seats, $F$ will assign the next seat to the largest party among them. This coincides with the behavior of Adams' method. Finally, once $a_\ell>0$ for all parties $P_\ell$, it is straightforward to see that $F$ behaves like Adams' method, since their $d$-sequences are identical from the second entry onward.
\end{proof}

\JUQVerifPoly*

\begin{proof}
	Fix a committee $W$. Clearly, there are only polynomially-many pairs of candidates $c$ and $d$ to check. Furthermore, there are polynomially-many possible sizes of $\voters^\prime$ to check, namely $n_d\leq n$. We show how to check, for each $(c,d)\in\candidates^2$ and $s\in[n_d]$, whether a counterexample for JUQ exists with candidate pair $(c, d)$ and with $|\voters^\prime|=s$.

	First, we check whether the condition on $\voters(c)$ holds (i.e., all voters in this set are above their upper quota). If not, then we are done. If yes, then we compute the quota associated to $s$ as $\nicefrac{(k\cdot s)}{n}$ and check, for every $i\in\voters(d)$, whether $|\app_i\cap ((W\cup\{d\})\setminus\{c\})|\leq\ceil{\nicefrac{(k\cdot s)}{n}}$ holds. If there exists at least $s$ voters in $\voters(d)$ such that the previous conditions holds, then we can use such voters to form $\voters^\prime$ and JUQ is failed. Otherwise, no counterexample with $\voters^\prime$ of size $s$ exists.
\end{proof}

\adamschar*
\begin{proof}
	First, observe that any refinement of Adams-AV satisfies apportionment upper quota (\Cref{prop:JUQ-apportionment} and \Cref{thm:JUQ-adams}). Fix $\boldsymbol{w}$ and suppose $F_{\boldsymbol{w}}$ satisfies apportionment upper quota. We will show that $F_{\boldsymbol{w}}$ refines Adams-AV. We will work in five steps, using five apportionment instances.
	\begin{enumerate}
		\item First, we show that $w_1=(1,0,\ldots)$.
		\item Second, we argue that $\boldsymbol{w}$ has at least two entries.
		\item Third, we prove that $w_2(s)>0$ for all $s\in\posNats$.
		\item Fourth, we demonstrate that $w_2(s)\geq \nicefrac{1}{(s-1)}$ for all $s>1$.
		\item Finally, we conclude that $w_2=(1,1,\nicefrac{1}{2},\nicefrac{1}{3},\ldots)$. 
	\end{enumerate}
	
	The above is sufficient to show that $F_{\boldsymbol{w}}$ refines Adams-AV. In the following, we will only argue about apportionment instances. Observe that any two $W$ and $W^\prime$ with $\boldsymbol{a}(W)=\boldsymbol{a}(W^\prime)$ yield the same $w$-score for any weight function $w$. Therefore, we will only speak about apportionment vectors, and ignore the underlying committee.
	
	First, fix $k=2$ and consider $n\in 2\posNats$ with $n>2$. Consider profile
	\[\approfile=(\underbrace{1, \ldots, 1}_{\nicefrac{n}{2}\text{ many}}, \nicefrac{n}{2}).\]
	Since $w_1(1)=1\geq w_1(2)$, we cannot have $a_{1+\nicefrac{n}{2}}=0$, as we could increase $\score_{w_1}(\approfile, k)$ by removing one seat from some party $\ell$ with $\ell\in[\nicefrac{n}{2}]$ and giving it to party $1+\nicefrac{n}{2}$. Moreover, $\lceil q(P_\ell)\rceil=1$ for all parties $\ell\in[1+\nicefrac{n}{2}]$. By apportionment upper quota, we must then have $a_{1+\nicefrac{n}{2}}=1$. Therefore,
	\[\nicefrac{n}{2}\cdot(w_1(1)+w_1(2)) \leq (\nicefrac{n}{2}+1)\cdot w_1(1) \implies \nicefrac{n}{2}\cdot w_1(2) \leq w_1(1).\]
	Since $w_1$ does not depend on $n$ (which can be arbitrarily large), we obtain $w_1(2)=0$.
	
	Next, assume $\approfile=(1,\,2)$ and $k=3$. If $\boldsymbol{w}=(w_1)$, then some $W$ with $\boldsymbol{a}(W)=(2,\,1)$ is among the winning committees, which violates apportionment upper quota. Thus, $\boldsymbol{w}$ has at least two entries. Note that we can assume $w_1\neq w_2$ without loss of generality, which implies $w_2(2)>0$.
	
	Now, assume that $w_2(s)=0$ for some $s\in\posNats$. Pick $s^*$ to be the smallest such $s$, that is, $s^*=\min\{s\in\posNats\colon w_2(s)=0\}$. Observe that $s^* > 2$. Let $k=2s^*-2$. Construct profile $\approfile=(s^*,\,s^*-2)$ with $n=2s^*-2$ voters. Note that $\lceil q(P_1)\rceil=s^*$ and $\lceil q(P_2)\rceil=s^*-2$. However, since $w_1(2)=w_2(s^*)=0<w_2(s^*-1)$, any winning committee must have $a_1=a_2=s^*-1$, violating apportionment upper quota. Therefore, $w_2(s)>0$ for every $s\in\posNats$. 
	
	From now on, we can safely assume that $w_2(1)=w_2(2)=1$. Fix $s > 2$. Pick an $\epsilon\in\rationals\cap(0, 1)$ where $\epsilon$ is arbitrarily small and some $h\in\posNats$ such that $h\cdot \epsilon\in\posNats$. Let $g=\lceil \nicefrac{1}{\epsilon} \rceil$ and observe that $1\leq g\epsilon$. Let $k=g(s-1)+2$ and $n=h+gh(s-1+\epsilon)$. Consider profile
	\[\approfile=(h,\,\underbrace{h(s-1+\epsilon),\,\ldots,\,h(s-1+\epsilon)}_{g\text{ many}}).\]
	Then, we have
	\[q(P_1) = \frac{h}{h+gh(s-1+\epsilon)}\left[g(s-1)+2\right] \leq 1\quad \impliedby \quad 1 \leq g\epsilon\]
	which holds by definition of $g$. Similarly, for $1<\ell\leq g+1$,
	\[q(P_\ell) = \frac{h(s-1+\epsilon)}{h+gh(s-1+\epsilon)}\left[g(s-1)+2\right] > s-1 \quad \impliedby \quad s-1+2\epsilon > 0,\]
	which holds because $s>2$ and $\epsilon<1$. Since any winning $W$ must be $w_1$-optimal, given $k\geq g+1$, we must have $a_\ell(W)>0$ for all $\ell\in[g+1]$. Therefore, $a_1=1$, as otherwise we would violate apportionment upper quota. Now, consider two committees $W$ and $W^\prime$ such that $\boldsymbol{a}(W)=(1,s,s-1,\ldots,s-1)$ and $\boldsymbol{a}(W^\prime)=(2,s-1,\ldots,s-1)$. First, observe that $w_2$-score of $W$ must be higher or equal than the $w_2$-score of any $W^\dagger$ with $a_1(W^\dagger)=1$. This holds because $w_2$ is (by definition) weakly decreasing and every two parties $\ell,\ell^\prime\in[g+1]\setminus\{1\}$ have $|P_\ell|=|P_{\ell^\prime}|$. Hence, we can achieve maximal $w_2$-score (among committees $W^\dagger$ with $a_1(W^\dagger)=1$) by spreading the seats over the parties in $[g+1]\setminus\{1\}$ as evenly as possible, which is what $W$ achieves. Hence, since $W^\prime$ violates apportionment upper quota, its $w_2$-score must be strictly smaller than the $w_2$-score of $W$, yielding
	\begin{multline*}
		h(w_2(1)+w_2(2))+gh(s-1+\epsilon)\sum_{i=1}^{s-1}w_2(i)<
		\\ h\cdot w_2(1) + h(s-1+\epsilon)\cdot w_2(s) + gh(s-1+\epsilon)\sum_{i=1}^{s-1}w_2(i),
	\end{multline*}
	which in turn implies
	\[\frac{w_2(2)}{s-1+\epsilon} < w_2(s).\]
	Since $w_2(2)=1$ and this holds for any $s>2$ and arbitrarily small $\epsilon$ we get that $w_2(s)\geq \nicefrac{1}{(s-1)}$ for all $s>1$. 
	
	Finally, we will show that $w_2(s) = \nicefrac{1}{(s-1)}$ for all $s>1$ by induction on $s$. The base case $(s=2)$ is trivial, as we already know $w_2(2)=1$. Hence, fix $s>2$ and consider the induction hypothesis
	\begin{center} ``$w_2(i)=\nicefrac{1}{(i-1)}$ for $i\in[s-1]\setminus\{1\}$.''\end{center}
	Pick an $\epsilon\in\rationals\cap(0, 1)$ where $\epsilon$ is arbitrarily small and some $h\in\posNats$ such that $h\cdot \epsilon\in\posNats$. Let
	\[g = \left\lceil \frac{s-1}{\epsilon} \right\rceil.\]
	Fix $k=g+s$ and consider the profile with $n=h(g+s-1-\epsilon)$ voters
	\[\approfile=(\underbrace{h,\,\ldots,\,h}_{g\text{ many}},\,h(s-1-\epsilon)).\]
	Observe that, for any $\ell\in[g]$, we have
	\[q(P_\ell) = \frac{h}{h(g+s-1-\epsilon)}(g+s) > 1 \impliedby 1+\epsilon > 0, \]
	which holds by definition of $\epsilon$. Next,
	\[q(P_{g+1}) = \frac{h(s-1-\epsilon)}{h(g+s-1-\epsilon)}(g+s) \leq s-1\impliedby g+1\geq \nicefrac{(s-1)}{\epsilon},\]
	which again by definition of $g$. Moreover, observe that since any winning committee must be $w_1$-optimal and $k>g+1$, every party gets assigned at least one seat. Consequently, we have $k-(g+1)=s-1$ seats left to assign. Let $p=s-1-a_{g+1}$. We now show that $p\leq 0$, i.e. that the $(g+1)$-th party receives at least $s-1$ seats. Indeed, suppose $p>0$. Since
	\[ p\leq s-1\leq \nicefrac{s-1}{\epsilon}\leq g,\]
	we can achieve maximal $w_2$-score by any apportionment of form
	\[(\underbrace{2,\,\ldots,\,2}_{(p+1)\text{ many}},\,\underbrace{1,\,\ldots,\,1}_{(g-p-1)\text{ many}},\,s-1-p).\]
	However, the $w_2$-score of the above is strictly smaller than the $w_2$-score we would obtain with apportionment $(2,\,1,\,\ldots,\,1,\,s-1)$. Indeed,
	\begin{multline*}
		gh\cdot w_2(1) + h(p+1)w_2(2) + h(s-1-\epsilon)\sum_{i=1}^{s-1-p} w_2(i) < \\
		gh\cdot w_2(1) + h\cdot w_2(2) + h(s-1-\epsilon)\sum_{i=1}^{s-1} w_2(i),
	\end{multline*}
	which for positive $p$ holds if
	\begin{equation}
		p\cdot w_2(1)=p < (s-1-\epsilon) \sum_{i=s-p}^{s-1} w_2(i).\label{eq:proof_step}
	\end{equation}
	Now, by the induction hypothesis, the smallest addend involved in the summation on the right-hand side is
	\[\frac{s-1-\epsilon}{s-2}>1,\]
	because $\epsilon<1$. Therefore
	\[p = \sum_{i=s-p}^{s-1} 1 < (s-1-\epsilon) \sum_{i=s-p}^{s-1} w_2(i),\]
	which shows that Equation~\eqref{eq:proof_step} holds. Hence, $a_{g+1}\geq s-1$. However, since $a_{g+1} = s$ violates apportionment upper quota, we must have
	\[gh\cdot w_2(1)+h\cdot w_2(2) + h(s-1-\epsilon)\sum_{i=1}^{s-1}w_2(i) > gh\cdot w_2(1) + h(s-1-\epsilon)\sum_{i=1}^sw_2(i),\]
	which in turn implies
	\[\frac{w_2(2)}{s-1-\epsilon} > w_2(s).\]
	Since $w_2(2)=1$ and this holds for arbitrarily small $\epsilon$, we have that $w_2(s)\leq\nicefrac{1}{(s-1)}$, and thus, by the previous arguments in this proof, $w_2(s)=\nicefrac{1}{(s-1)}$. This closes the induction step, and we are done.
\end{proof}

\JUQChains*

\begin{proof}
	Clearly, no JUQ swap can decrease the number of covered voters in a committee, as any overrepresented voter must have satisfaction at least $2$. Moreover, given a committee $W$ that covers the maximal number of voters, by inspecting the proof of Theorem~\ref{thm:JUQ-adams}, we can observe that JUQ swaps always increase the $w_2$-score of $W$, where $w_2=(1,1,\nicefrac{1}{2},\nicefrac{1}{3},\ldots)$. Since there are finitely many committees, it follows that sequences of JUQ swaps must always terminate. 
\end{proof}

\UQEREJR*

\begin{proof}
    It is easy to see that the initial committee $W_0=C$ satisfies EJR+, as there are no candidates outside the committee that can witness an EJR+ violation. We show by induction that UQER maintains EJR+.
    
	Suppose removing a candidate $c_j$ is the first time we encounter an EJR+ violation. As we started from the set of all candidates $C$, this violation must involve some previously removed candidate $c_x$, $x<j$ and voter group $N'\subseteq N(c_x)$. This means that $|N'|\leq |N(c_x)|\leq |N(c_j)|$ from the definition of our rule, and thus $q(N')\leq q_{c_x}\leq q_{c_j}$. 
		
	We know that $|A_i\cap W_{j-1}|>\ceil{q_{c_j}}$ for all $i\in N(c_j)$, as $c_j$ violated UQ in $W_{j-1}$. We also know that $|A_i\cap (W_{j-1}\setminus \{c_j\})|<\floor{q(N')}$ for all $i\in N'$, because of our EJR+ violation assumption.
		
	Suppose $N(c_j)\cap N'\neq \emptyset$ and let $i^* \in N(c_j)\cap N'$. From the above, $\floor{q(N')}\geq |A_{i^*}\cap W_{j-1}|>\ceil{q_{c_j}}\geq \ceil{q(N')}$, which is a contradiction. Thus, $N(c_j)\cap N'=\emptyset$. This means that $|A_i\cap W_{j-1}|=|A_i\cap (W_{j-1}\setminus \{c_j\})|<\floor{q(N')}$ for all $i\in N'$, and thus $c_x$ with voter group $N'$ would have violated EJR+ in $W_{j-1}$, which contradicts our assumption that $W_j$ is the first time we encounter an EJR+ violation.
		
    Thus every candidate set from $\{W_0,W_1,\dots, W_{j^*}\}$ satisfies EJR+.
\end{proof}

\nearquotaequiv*

\begin{proof}
	The statement follows immediately from the definition via algebraic manipulation.
\end{proof}

\JNQApportionment*

\begin{proof}
If a committee satisfies JNQ, then it clearly satisfies apportionment near quota. Let us prove the inverse direction. 

Suppose that $W$ satisfies apportionment near quota. Assume towards a contradiction that $c$, $d$ and $\voters^\prime$ are witnesses to a JNQ violation. Let $i$ be the party of voters approving $c$ and $j$ the party of voters approving $j$. 

First, assume $i=j$. This implies $\voters^\prime\subseteq\voters(c)$ and hence $q(\voters^\prime)\leq q_c$. However, the inequalities that define JNQ reduce to $a_i>q_c$ and $q(\voters^\prime)>a_i$: contradiction.

Hence, assume $i\neq j$. In this case, the first condition of JNQ reduces to $a_i-q_c>\nicefrac{1}{2}$ and the second to $q(\voters^\prime)-a_j>\nicefrac{1}{2}$. Since $\voters^\prime\subseteq\voters(d)$, we have $q(\voters^\prime)\leq q_d$ and hence $q_d-a_j<\nicefrac{1}{2}$. This is a violation of apportionment near quota, completing the proof.
\end{proof}

\slav*

\begin{proof}
    Observe that the JNQ violation implies $2q_c <2|A_i\cap W|-1$ for all voters $i\in N(c)\setminus N(d)$
    and
    $2q(N')>1+2|A_i\cap W|$ for all voters $i\in N'\setminus N(c)$.
    
    Assume for a contradiction that we have a JNQ violation witnessed by $c\in W$, $d\in C\setminus W$ and $N'\subseteq \voters(d)$ that changes the SLAV score by:
    \begin{multline*}
        \score_w(\approfile, W \cup \{d\} \setminus \{c\}) - \score_w(\approfile, W)
        \ge \sum_{i \in N(c)\setminus N'} - \frac{1}{2 \lvert A_i \cap W \rvert - 1} + 
        \sum_{i \in N'\setminus N(c)}  \frac{1}{2 \lvert A_i \cap W \rvert + 1}> \\
        -\frac{|N(c)\setminus N'|}{2q_c}+\frac{|N'\setminus N(c)|}{2q(N')}
    \end{multline*}

    We consider two cases. If $N'\cap \voters(c)=\emptyset$, then:
    \begin{multline*}
	    \score_w(\approfile,W \cup \{d\} \setminus \{c\}) - \score_w(\approfile,W)
	    >-\frac{|N(c)\setminus N'|}{2q_c}+\frac{|N'\setminus N(c)|}{2q(N')}=\\
	    -\frac{|N(c)|}{2q_c}+\frac{|N'|}{2q(N')}
	    =-\frac{n}{2k}+\frac{n}{2k}
	    =0
    \end{multline*}
    
    Suppose instead that there exists some voter $i\in N'\cap \voters(c)$. This implies $q_c <q(N')$. Then:

    \begin{multline*}
        \score_w(\approfile,W \cup \{d\} \setminus \{c\}) - \score_w(\approfile,W)
        >-\frac{|N(c)\setminus N'|}{2q_c}+\frac{|N'\setminus N(c)|}{2q(N')}=\\
        -\frac{|N(c)|}{2q_c}+\frac{|\voters(c)\cap N'|}{2q_c}+\frac{|N'|}{2q(N')}-\frac{|\voters(c)\cap N'|}{2q(N')}
        =|\voters(c)\cap N'|\bigg(\frac{1}{2q_c}-\frac{1}{2q(N')}\bigg)
        >0
    \end{multline*}

    Hence, the SLAV score would strictly increase by swapping from candidate $c$ to $d$ and therefore SLAV satisfies JNQ.
\end{proof}

\JNQChain*

\begin{proof}
    The proof is analogous to the proof of Proposition~\ref{prop:JUQ-chain} by referencing the proof of Theorem~\ref{thm:SLAV_JNQ} instead of Theorem~\ref{thm:JUQ-adams}.
\end{proof}

\JNQVerifPoly*

\begin{proof}
 This proof is essentially analogous to the proof of \Cref{prop:juq-verif-poly}. The only difference is that, for every voter in $i\in\voters(d)$, we check whether $\nicefrac{(k\cdot s)}{n}-  \lvert A_i \cap W\rvert> \nicefrac{1}{2} \cdot \indicator[c \notin A_i]$ holds instead.
\end{proof}

\uqerelimination*

\begin{proof}
    Consider the multi-winner instance with $n=15$, $k=10$ depicted in Table~\ref{tab:UQER_depletion}.

\begin{table}[h!]
    \small
    \setlength{\tabcolsep}{3pt}
    \centering
    \begin{tabular}{|c|ccc:ccc:ccc:c:c|}
         \hline
         Candidate  & $A_1$ & $A_2$ & $A_3$ & $A_4$ & $A_5$ & $A_6$ & $A_7$ & $A_8$ & $A_9$ & $A_{10}-A_{12}$ & $A_{13}-A_{15}$ \\
         \hline
         $\hat{c}_1,\hat{c}_2$ &  \cmark & \cmark & \cmark &  & & & & & & & \\
         $c_1,c_2$ & \cmark & & &  & &  & & & & \cmark & \\
         $c_3,c_4$ & & \cmark & &  & &  & & & & \cmark & \\
         $c_5,c_6$ & & & \cmark &  & &  & & & & \cmark & \\
         \hdashline
         $\hat{c}_3,\hat{c}_4$ & & & & \cmark & \cmark & \cmark & & & & & \\
         $c_7,c_8$ & & & & \cmark & &   & & & & \cmark & \\
         $c_9,c_{10}$ & & & & & \cmark &   & & & & \cmark & \\
         $c_{11},c_{12}$ & & & & & & \cmark   & & & & \cmark & \\
         \hdashline
         $\hat{c}_5,\hat{c}_6$ & & & & & & & \cmark & \cmark & \cmark & & \\
         $c_{13},c_{14}$ & & & & & & &  \cmark & & &  \cmark & \\
         $c_{15},c_{16}$ & & & & & & & & \cmark & &  \cmark & \\
         $c_{17},c_{18}$ & & & & & & & & &  \cmark &   \cmark & \\
         \hdashline
         $c^*_1,c^*_{2}$ &  & & & & & & & & & & \cmark  \\
         \hline
    \end{tabular}
    \caption{Example instance for \Cref{prop:UQER_elimination}.}
    \label{tab:UQER_depletion}
\end{table}

In order to obtain $W_{j^*}$, UQER removes $\hat{c}_1-\hat{c}_{6}$, as these are the UQ-violating candidates with fewest supporters, and terminates.

From the remaining candidates, in order to satisfy EJR+\footnote{Or even the weaker proportionality notion of EJR.} we must select:

\begin{itemize}
    \item At least 3 candidates from each of $\{c_1,\dots, c_6\}$, $\{c_7,\dots, c_{12}\}$ and $\{c_{13},\dots, c_{18}\}$. This holds as each group of three voters among the first nine must have one voter who approves of at least two candidates in the outcome, and another who approves of at least one.
    \item Both $c^*_1$ and $c^*_2$, as the last three voters together deserve at least $2$ candidates according to EJR+.
\end{itemize}

This is $11>k$ candidates. Thus $W_{j^*}$ need not contain a subset of size $k$ that satisfies EJR+.
\end{proof}

Before proving the next theorem, we define the Method of Equal Shares~\citep{PeSk20b}.

\begin{definition}
	The Method of Equal Shares (MES) works sequentially in rounds. We begin with the empty committee, $W_0=\emptyset$, and at each round $r$ we add a candidate to $W_{r-1}$ to obtain $W_r$. Initially, all voters are assigned a budget of $\nicefrac{k}{n}$. Let $b_r(i)$ be the budget of voter $i$ before round $r$ (hence, $b_1(i)=\nicefrac{k}{n}$ for all $i\in\voters$). In round $r$, we consider the set $\candidates_r\subseteq\candidates$ of all candidates in $\candidates\setminus W_r$ such that $\sum_{i\in\voters(c)}b_r(i)\geq1$. If $\candidates_r$ is empty, the rule terminates and returns $W_r$. Otherwise, we compute $\alpha(c)$ for each $c\in\candidates_r$ as
\[ \alpha(c) = \min\left\{\alpha\in\reals\mid \sum_{i\in\voters(c)}\min(\alpha,\,b_r(i)) = 1 \right\}. \]
We select a candidate $c^*$ with minimal $\alpha(c^*)$ and set $b_{r+1}(i)=\max(0, b_r(i)-\alpha(c^*))$ if $i\in\voters(c^*)$ and $b_{r+1}(i)=b_r(i)$ otherwise.

Note that MES might terminate in less than $k$ rounds, and thus return a committee with fewer than $k$ candidates. To enforce output committees to have size exactly $k$, one needs to extend MES with a so-called \emph{completion method}~\cite[Chapter 2]{LaSk22a}.
\end{definition}

\mesjuqswaps*

\begin{proof}

We structure our counterexample in 3 parts. First, we construct a multi-winner instance, then derive the outcome of MES (with arbitrary completion), and finally, consider sequences of JUQ swaps starting from that outcome.

\paragraph{Constructing an instance} We construct a counterexample instance as follows. For notational convenience, let $x=31$. Moreover, let $k=140$. Set the number of voters to be $n=6xk$, which notably means that the quota of a group of $6x$ voters will be exactly $1$.

To assist with our argument, consider two multi-winner instances, the ``single gadget'' in \Cref{tab:destabilizing_gadget} and the ``double gadget'' in \Cref{tab:destabilizing_gadget2}. Observe that in the former MES selects all $6$ candidates, which is an exhaustive outcome. Here, candidates $c_1$, $c_2$ and $c_3$ all violate UQ. Similarly, MES selects all candidates in the double gadget, and $c'_1-c'_6$ all witness UQ violations.

We add 15 copies of the single gadget to our instance (giving $15\cdot6\cdot 6x$ voters), as well as one double gadget, with all gadgets being disjoint, i.e. a voter from any gadget approves only candidates from that gadget. 
Call candidates $c'_1-c'_6$ of the double gadget ``Type A'' and $c'_7-c'_{18}$ ``Type B''. Call the voters sets in the last 24 columns of the double gadget $\{S_1,\dots,S_{12}\}$, as indicated in \Cref{tab:destabilizing_gadget2}.

\begin{table}[h!]
    \small
    \setlength{\tabcolsep}{3pt}
    \centering
    \begin{tabular}{|c|ccc:ccc|}
        \hline
        Candidates & $6x$ & $6x$ & $6x$ & $6x$ & $6x$ & $6x$ \\
        \hline
        $c_1-c_3$ & \cmark & \cmark & \cmark & & & \\
        \hdashline
        $c_4$ & \cmark & & & \cmark & & \\
        $c_5$ & & \cmark & & & \cmark & \\
        $c_6$ & & & \cmark & & & \cmark \\
        \hline
    \end{tabular}
    \caption{Single Gadget: Multi-winner instance with $n=36x$ and $k=6$. ``$6x$'' in the column headers represents the number of voters with that approval profile.}
    \label{tab:destabilizing_gadget}
\end{table}

\begin{table}[h!]
    \footnotesize
    \setlength{\tabcolsep}{2pt}
    \centering
    \begin{tabular}{|c|cccccc:cccccc:cccccccccccc|}
        \hline
        & $R_1$ & $R_2$ & $R_3$ & $R_4$ & $R_5$ & $R_6$ & $R_7$ & $R_8$ & $R_9$ & $R_{10}$ & $R_{11}$ & $R_{12}$ 
        & $S_1$ & $S_2$ & $S_3$ & $S_4$ & $S_5$ & $S_6$ & $S_7$ & $S_8$ & $S_9$ & $S_{10}$ & $S_{11}$ & $S_{12}$ 
        \\
        \hline
        Candidates 
        & $3x$ & $3x$ & $3x$ & $3x$ & $3x$ & $3x$ & $3x$ & $3x$ & $3x$ & $3x$ & $3x$ & $3x$ 
        & $3x$ & $3x$ & $3x$ & $3x$ & $3x$ & $3x$ & $3x$ & $3x$ & $3x$ & $3x$ & $3x$ & $3x$ 
        \\
        \hline
        $c'_1-c'_3$ & \cmark & \cmark & \cmark & \cmark & \cmark & \cmark  & & & & & & & & & & & & & & & & & &  \\
        \hdashline
        $c'_4-c'_6$ & & & & & & & \cmark & \cmark & \cmark & \cmark & \cmark & \cmark  & & & & & & & & & & & &  \\
        \hdashline
        $c'_7$ & \cmark & & & & & & \cmark & & & & & & \cmark & & & & & & \cmark & & & & & \\
        $c'_8$ & & \cmark & & & & & & \cmark & & & & & & \cmark & & & & & & \cmark & & & &  \\
        $c'_9$ & & & \cmark & & & & & & \cmark & & & & & & \cmark & & & & & & \cmark & & & \\
        $c'_{10}$ & & & & \cmark & & & & & & \cmark & & & & & & \cmark & & & & & & \cmark & & \\
        $c'_{11}$ & & & & & \cmark & & & & & & \cmark & & & & & & \cmark & & & & & & \cmark &  \\
        $c'_{12}$ & & & & & & \cmark & & & & & & \cmark & & & & & & \cmark & & & & & & \cmark \\
        \hline
    \end{tabular}
    \caption{Double Gadget: Multi-winner instance with $n=72x$ and $k=12$. ``$3x$'' in the column headers represents the number of voters with that approval profile. We call the voter sets in the first $12$ columns $R_1$ through $R_{12}$, and those in the last $12$ columns $S_1$ through $S_{12}$.}
    \label{tab:destabilizing_gadget2}
\end{table}

We additionally add the following voter sets to the instance:

\begin{itemize}
    \item $\{T_1,\dots,T_{12}\}$, each containing $9x+1$ voters.
    \item $\{U_1,\dots,U_{12}\}$, each containing $9x+1$ voters.
    \item $Z$, containing $12x-24$ additional voters to ensure $n=6xk$.
\end{itemize}

This gives us a total of $140\cdot6x=26040$ voters.

Divide the $216x+24$ voters from $\{T_1,\dots,T_{12}\}$ and $\{U_1,\dots,U_{12}\}$ into $35$ disjoint groups of equal size (each containing $192$ voters), and call these groups $\{Y_1,\dots,Y_{35}\}$.\footnote{The choice of $x=31$ ensured that $216x+24$ was divisible by $35$.}

We then add the following candidates to the instance, including them in the approval sets of some voters from the double gadget, and some of the additional voters above:

\begin{itemize}
    \item ``Type C'' candidates:
    \begin{itemize}
    \item Candidates $\hat{c}_i$ for $1\leq i\leq 12$, approved by $S_i \cup T_i$.
    \item Candidates $\hat{c}_{i+12}$ for $1\leq i\leq 12$, approved by $S_i \cup U_i$.
    \end{itemize}
    \item ``Type D'' candidates:
    \begin{itemize}
    \item Candidates $\bar{c}_{i}$ for $1\leq i\leq 34$, approved by $Y_i$ and $Y_{i+1}$.
    \item Candidate $\bar{c}_{35}$, approved by $Y_{35}$ and $Y_{1}$.
    \end{itemize}
    \item Candidate $c^Z$, approved by the voters in $Z$.
\end{itemize}

\paragraph{MES outcome} Let $n_A,n_B,n_C,n_D$ and $q_A,q_B,q_C,q_D$ be the number of supporters and quotas of each candidate type. Observe the following:

\begin{itemize}
    \item Type A candidates are approved by $n_A=18x$ voters and thus have $q_A=3$
    \item Type B candidates are approved by $n_B=12x$ voters and thus have $q_B=2$
    \item Type C candidates are each approved by $12x+1$ voters, which means $2<q_C<3$
    \item Type D candidates are each approved by $18x>384>12x+1$ voters, and thus $2<q_D<3$ also, but $q_D>q_C$.
\end{itemize}

For the purposes of MES, we can treat each single gadget separately, as they are disjoint, and also disjoint from the (enhanced) double gadget. We can also treat the voters in $Z$ separately.
MES selects every candidate from every single gadget ($15\cdot6=90$ in total), as well as $c^Z$. From the enhanced double gadget, MES selects, in order:

\begin{itemize}
    \item All $6$ Type A candidates $\{c'_1,\dots,c'_6\}$, exhausting the budgets of voters in $\{R_1,\dots,R_{12}\}$.
    \item All $35$ Type D candidates $\{\bar{c}_1,\dots,\bar{c}_{35}\}$. At this point, the voters in $\{T_1,\dots,T_{12}\}$ and $\{U_1,\dots,U_{12}\}$ have purchased $35$ candidates, from their combined MES budgets of $(216x+24)/6x\approx 36.129$. Thus, by symmetry, each of the $24$ voter groups has less than $0.048$ budget remaining among its $9x+1$ members. Meanwhile, each voter group in $\{S_1,\dots,S_{12}\}$ has an MES budget of $0.5$ among its $3x$ members. We can conclude that no Type C candidate is affordable by its supporters.
    \item All $6$ Type B candidates $\{c'_7,\dots,c'_{12}\}$.
\end{itemize}

Thus, MES selects a total of $138$ candidates. To complete this committee to size $k=140$, we choose any two Type C candidates, and call the resulting set $W$.

\paragraph{JUQ swaps} First, observe that no Type B candidate currently constitutes a UQ violation for $W$, even if any one additional candidate is added to it. For instance, a committee in which $c'_7$ violates UQ, must also contain $\hat{c}_1,\hat{c}_7,\hat{c}_{13}$ and $\hat{c}_{19}$, as well as two of $c'_1-c'_3$, and two of $c'_4-c'_6$ in order to give all voters in $R_1,R_7,S_1$ and $S_7$ a satisfaction of $3$. In other words, it is necessary (but not sufficient) for at least 4 Type C candidates to be in a committee in order to involve any Type B candidate in a UQ violation.

In $W$, the only candidates that violate UQ are $\{c_1,c_2,c_3\}$ in each of the 15 single gadgets, as well as the Type A candidates from $\{c'_1,c'_2,c'_3\}$ and $\{c'_4,c'_5,c'_6\}$ in the double gadget. Observe that each of these candidates has exactly the same amount of supporters, $18x$, with each of those supporters obtaining a satisfaction of $4$ from committee $W$. Removing one candidate from any set of three candidates listed above leads to the other two no longer violating UQ. We assume our method of selecting JUQ swaps uses the profile of a candidate's supporter satisfactions to select which swap to satisfy. Thus, it cannot distinguish between these violations, so we will assume that the first candidate removed in a JUQ swap is $c'_1$. Each of the 22 remaining Type C candidates not selected by our completion method can be added to the outcome without violating upper quota. Observe that adding the third Type C candidate cannot bring any additional candidate over quota: Type C candidates can never constitute a UQ violation and a Type B candidate can only constitute a UQ violation once at least four Type C candidates have been added. Thus, we will assume that the second candidate removed in a JUQ swap is $c'_4$. 

As we continue to perform JUQ swaps we will remove 15 UQ-violating single gadget candidates and 5 of the 6 Type B candidates, in some order, while adding in the remaining 20 Type C candidates. Notably, only the last of these JUQ swaps will make the remaining Type B candidate violate UQ (let's say this is $c'_{12}$ WLOG). In addition to $c'_{12}$, the committee currently contains 5 candidates from each of the 15 single gadgets, $c^Z$, 4 Type A candidates $\{c'_2,c'_3,c'_5,c'_6\}$, all 24 Type C Candidates, and all 25 Type D candidates. Observe that $c'_{12}$ witnesses a JUQ violation with either $c'_1$ or $c'_4$, the Type A candidates we initially removed. Notably, we did not want to add a Type A candidate back into the committee as part of a JUQ swap until 5 of the 6 Type B candidates were removed. Performing either JUQ swap leads to a committee for which the other Type A candidate witnesses an EJR+ violation. Furthermore, no candidates in that committee violate UQ, and thus the sequence of JUQ swaps terminates in a committee violating EJR+.
\end{proof}

\section{Additional Examples and Results}\label{app:additional}

In this section, we report additional examples and results that did not fit into the main body.
\subsection{Adams-AV Alternative Formulation}\label{app:adams:thiele}
\begin{observation}
	For instances $(\approfile, k)$ with $n$ voters Adams-AV is equivalent to the Thiele rule with vector $w = (n\cdot k, 1, \frac{1}{2}, \frac{1}{3},\dots)$.
\end{observation}
\begin{proof}
	Let $W$ and $W'$ be two committees. To show that equivalence, we show that if $W$ has a strictly lexicographically higher ``Adams-AV score'' than $W'$ if and only if $W$ also has a strictly higher score according to $w$. As in the other proofs we let $w_1 = (1,0,0,\dots)$ be the first score vector of Adams-AV and $w_2 = (1,1,\frac{1}{2}, \frac{1}{3}, \dots)$ the second.
	
	First, we show that if $\score_{w_1}(\approfile, W) > \score_{w_1}(\approfile, W')$ then  $\score_{w}(\approfile, W) > \score_{w}(\approfile, W')$. 
	
	Let $t = \score_{w_1}(\approfile, W)$ and $t' = \score_{w_1}(\approfile, W')$. That is, $t$ voters approve a candidate in $W$ while $t'$ voters approve a candidate in $W'$ with $t > t'$. As a consequence the $w$-score of $W$ is at least $t \cdot n \cdot k$. On the other hand, the $w$-score of $W'$ is at most $t' \cdot n\cdot k + t' \cdot H_k < t' \cdot n\cdot k + (n-1)\cdot k < (t' + 1) \cdot n \cdot k \le t \cdot n \cdot k$, where $H_k$ is the $k$-th harmonic number. Thus, in this case $\score_{w}(\approfile, W) > \score_{w}(\approfile, W')$. 
	
	Secondly, we observe the following: based on the definition of $w$ for any committee $W''$ it holds that $\score_{w}(\approfile, W'') = \score_{w_2} (\approfile, W'') + \score_{w_1}( \approfile, W'') \cdot (n\cdot k - 1)$ as $\score_{w_1}( \approfile, W'') \cdot (n\cdot k - 1)$ is precisely the difference in the ``first coordinate'' between $w_2$ and $w$. 
	
	Thus, for our two committees $W$ and $W'$ if $\score_{w_1}(\approfile, W) = \score_{w_1}(\approfile, W')$ we immediately obtain that $\score_{w_2}(\approfile, W) > \score_{w_2}(\approfile, W')$ if and only if $\score_{w}(\approfile, W) > \score_{w}(\approfile, W')$ as $\score_{w_1}( \approfile, W'') \cdot (n\cdot k - 1)$ is the same for both committees.
	
	Consequently Adams-AV and the $w$-Thiele method are equivalent in this instance.
\end{proof}

\subsection{Swap Dynamics}\label{app:swap-dynamics-extra}

In the main body of the paper (Section~\ref{sec:JUQ-swap-dynamics}), we have shown that it is possible to remove (via JUQ swaps) and then add back (again, via JUQ swaps) the same candidate. Here, we show that the above problem persists even assuming some clever schemata on how to pick which JUQ swap to perform next. We present two such schemata.

Firstly, what if we always choose a swap $(c,d)$ such that $c$ has the least supporters? Consider the instance in \Cref{tab:swap_2}. Suppose we start with the following size $k$ committee: 
$$\{c_1,c_2,c_3,c_4,c_6,c_7,c_8,c_{12},\dots, c_{19}\}.$$

Notably, $c_5$ is not in this committee, while candidates in $\{c_1, c_2, c_3, c_{12}, \ldots, c_{19}\}$ have all their (respective) supporters above upper quota. We can replace $c_3$ (a minimally-supported, upper-quota-violating candidate) with, say, $c_5$. After this change, candidates in $\{c_6,c_7,c_8\}$ are over their upper quotas. We can replace each of these with a candidate from $\{c_9,c_{10}, c_{11}\}$. Note that now we can add $c_3$ back (e.g., by swapping it for $c_{19}$). 

Secondly, what if we consider some kind of condition on $d$ instead (or in addition)? If that condition relates to choosing a $d$ with the fewest or most supporters, then an analogous argument to the above applies.

\begin{table}[h!]
\small
    \centering
    \setlength{\tabcolsep}{3pt}
    \begin{tabular}{ccccccccccccc}
	    \toprule
	    Candidate  & $A_1$ & $A_2$ & $A_3$ & $A_4$ & $A_5$ & $A_6$ & $A_7-A_9$ & $A_{10}-A_{15}$ \\
         \midrule
         $c_1-c_3$ &  \cmark & \cmark & \cmark \\
         $c_4,c_5$ & & & & \cmark & \cmark & \cmark \\
         $c_6$     & \cmark & & & \cmark  \\
         $c_7$     & & \cmark & & & \cmark  \\
         $c_8$     & & & \cmark & & & \cmark  \\
	 $c_9-c_{11}$ & & & & & & & \cmark  \\
	 $c_{12}-c_{19}$ & & & & & & & & \cmark  \\
	 \bottomrule
    \end{tabular}
    \caption{Instance with $k=n=15$.}
    \label{tab:swap_2}
\end{table}

\subsection{Relationships with Other Proportionality Notions}\label{app:additional-more-lower-quota}

We define some additional proportionality notions for multi-winner voting.

\begin{definition}[\citep{PPS21a}]
	Given an integer $\ell\in\posNats$ and a set of candidates $\candidates^\prime\subseteq\candidates$, a group of voters $\voters^\prime\subseteq\voters$ is weakly $(\ell,\candidates^\prime)$-cohesive if $|\voters^\prime|\geq|\candidates^\prime|\cdot\nicefrac{n}{k}$ and if for all $i\in\voters^\prime$ it holds that $|A_i\cap \candidates^\prime|\geq\ell$.

	A committee $W$ satisfies \emph{fully justified representation} if for every $\ell\in\posNats$, set of candidates $\candidates^\prime\subseteq\candidates$, and weakly $(\ell,\candidates^\prime)$-cohesive group of voters $\voters^\prime$ there is some $i\in\voters^\prime$ with $|A_i\cap W|\geq\ell$.
\end{definition}

\begin{definition}[\citep{SFF+17a}]
	A committee $W$ satisfies \emph{perfect representation} if $\voters$ can be partitioned in $k$ disjoint groups $\left(\voters_i\right)_{i\in[k]}$ of size $\nicefrac{n}{k}$ and if we can assign to each group $\voters_i$ a distinct candidate $c_i\in W$ such that $c_i\in\bigcap_{j\in\voters_i} A_j$.
\end{definition}

\begin{definition}[\citep{PeSk20b}]
 A committee is \emph{priceable} if there exists a per-voter budget $b\in\posReals$ and a payment function $p_i\colon\candidates\rightarrow[0, 1]$ for each $i\in\voters$ such that:
 \begin{enumerate}
	 \item $\sum_{c\in\candidates}p_i(c)\leq b$ for each $i\in\voters$;
	 \item $p_i(c)=0$ for each $i\in\voters$ and $c\not\in A_i$; 
	 \item $\sum_{i\in\voters}p_i(c)=\indicator[c\in W]$;
	 \item $\sum_{i\in\voters(c)}\left(b-\sum_{d\in W}p_i(d)\right)\leq1$ for each $c\not\in W$.\qedhere
 \end{enumerate}
\end{definition}

Note that there are several competing notions of priceability defined in the literature. For instance, \citet{BrPe24a} use a notion where the per-voter budget is set to $b=\nicefrac{k}{n}$. We do not impose such a restriction here.

\begin{proposition}
	JUQ is incompatible with (i) fully justified representation, (ii) perfect representation, and, if we additionally require exhaustiveness, (iii) priceability.\label{prop:other-lower-quota-axioms}
\end{proposition}

\begin{proof}
    (i) Consider the instance in \Cref{tab:JUQ_FJR}. $A_{1}-A_{12}$ are 3-cohesive over $c_1-c_6$. Thus, without loss of generality, a committee satisfying FJR must contain $c_1$, $c_3$ and $c_4$. However, this leads to $c_4$, with $q_{c_4}=2$ violating UQ. Furthermore, for any committee of size $k\leq 10$ there will always be an unselected candidate $c\in \{c_7,\dots, c_{14}\}$, such that $(c_4,c)$ is a JUQ violation.
\begin{table}[h!]
    \small
    \setlength{\tabcolsep}{3pt}
    \centering
    \begin{tabular}{ccccccccccccc}
	    \toprule
        Candidate  & $A_1-A_4$ & $A_5-A_8$ & $A_9-A_{12}$ & $A_{13}$ & \dots & $A_{20}$ & \\
        \midrule
        $c_1$ &  \cmark & \cmark \\
        $c_2$ &  & \cmark & \cmark \\
        $c_3$ &  \cmark & & \cmark \\
        $c_4$ &  \cmark \\
        $c_5$ &  & \cmark \\
        $c_6$ &  & & \cmark \\
        $c_7$ &  & & & \cmark \\
        \dots & & & & & \dots \\
        $c_{14}$ & & & & & & \cmark \\
	    \bottomrule
    \end{tabular}
    \caption{Example instance with $k=10$.}
    \label{tab:JUQ_FJR}
\end{table}

    (ii) Consider the instance in \Cref{tab:JUQ_Perf_Rep}. A committee satisfying perfect representation exists and must contain $c_1$, $c_2$ and two of $\{c_3,c_4,c_5\}$. Without loss of generality, consider $\{c_1,c_2,c_3,c_4\}$. Then, $(c_2,c_5)$ witnesses a JUQ violation, as $q_{c_2}=1$ and $\ceil{q_{c_5}}=2$.
\begin{table}[h!]
    \small
    \setlength{\tabcolsep}{3pt}
    \centering
    \begin{tabular}{ccccccccccccc}
	    \toprule
        Candidate  & $A_1$ & $A_2$ & $A_3$ & $A_4$ & $A_5$ & $A_6$ & $A_7$ & $A_8$ & $A_9$ & $A_{10}$ & $A_{11}$ & $A_{12}$ \\
        \midrule
        $c_1$ & \cmark & \cmark & \cmark & \cmark & \cmark & \cmark \\
        $c_2$ & \cmark & \cmark & \cmark \\
        $c_3$ & & & & & & & \cmark & \cmark & \cmark & \cmark \\ 
        $c_4$ & & & & & & & \cmark & \cmark & & & \cmark & \cmark \\
        $c_5$ & & & & & & & & & \cmark & \cmark & \cmark & \cmark \\ 
	    \bottomrule
    \end{tabular}
    \caption{Example instance with $k=4$.}
    \label{tab:JUQ_Perf_Rep}
\end{table}

    (iii) Consider the instance depicted in \Cref{tab:JUQ_Priceability} and consider any committee $W$ of size $\lvert W\rvert = 3$. We will show that $W$ either violates priceability or JUQ.
    
    First, we consider the case that $W$ contains $c_5$ and at least one of $c_3$ and $c_4$ (without loss of generality $c_4$). Then, we observe that $\lceil q_{c_4}\rceil = 1$ and thus $c_4$ together with the unselected candidate among $c_1$ and $c_2$ witnesses a JUQ violation. 
    
    Next, assume that either $c_1$ or $c_2$ is contained in $W$. Then, if $W$ were priceable voters $1$ and $2$ must have a budget of at least $1$. If $W$ additionally contains $c_5$, we know by averaging that one of the groups $\{3,4\}$ or $\{5,6\}$ must have a leftover budget of $1.5$ therefore violating the fourth condition of priceability together with the candidate they commonly approve. If on the other hand, the committee contains either $c_3$ or $c_4$ the other of the two candidates could also be afforded by its supporters. 
    
    Hence, $W$ either violates priceability or JUQ.
\begin{table}[h!]
    \small
    \setlength{\tabcolsep}{3pt}
    \centering
    \begin{tabular}{ccccccccccccc}
	    \toprule
        Candidate  & $A_1$ & $A_2$ & $A_3$ & $A_4$ & $A_5$ & $A_6$ \\
        \midrule
        $c_1$ & \cmark \\
        $c_2$ & & \cmark \\
        $c_3$ & & & \cmark & \cmark \\
        $c_4$ & & & & & \cmark & \cmark \\
        $c_5$ & & & \cmark & \cmark & \cmark & \cmark \\
	    \bottomrule
    \end{tabular}
    \caption{Example instance with $k=3$.}
    \label{tab:JUQ_Priceability}
\end{table}
\end{proof}

\section{Counterexamples for Other Rules Satisfying JUQ or JNQ}\label{app:counterexamples}

In this section, we present several counterexamples for familiar multi-winner voting rules for JUQ and JNQ. We tested all the rules present in the \texttt{abcvoting} Python package~\citep{joss-abcvoting} at the time of writing, as well as some additional rules of interest. For the sake of simplicity, for the rules implemented in the \texttt{abcvoting} package, we refer to them by their name in the package.\footnote{See this link for more information: \url{https://abcvoting.readthedocs.io/en/latest/intro-abcrules.html} (accessed: 2026-02-05.)} For the other rules, we provide references. 

\subsection{JUQ}\label{app:counterexamples:JUQ}

\begin{description}
	
	\item[\textbf{Example 1}] 
	$m=6$, $n=10$, $k=2$.
	\[
	\begin{aligned}
		A_1&=\{a\},\\
		A_2&=\{b,c,d\},\\
		A_3&=\{b,c\},\\
		A_4&=\{b\},\\
		A_5&= \{e\},\\
		A_6&=\{b,c\},\\
		A_7&= \{f\},\\
		A_8&=\{d\},\\
		A_9&=\{b,c\},\\
		A_{10}&=\{b\}.
	\end{aligned}
	\]
	In this example $\{b,c\}$ does not satisfy JUQ, as we can swap $c$ with $a$. The rules AV, SAV, PAV, SLAV, GEOM2, SEQPAV, REVSEQPAV, SEQSLAV, SEQPHRAGMEN, MINIMAXPHRAGMEN, MAXIMIN-SUPPORT, MINIMAXAV, EQUAL-SHARES (in all its variants), and EPH select $\{b,c\}$ in this instance.
	
	\item[\textbf{Example 2}] 
	$m=4$, $n=4$, $k=3$.
	\[
	\begin{aligned}
		A_1&=\{b,c,d\},\\
		A_2&=\{a,d\},\\
		A_3&=\{b,c,d\},\\
		A_4&=\{b,c,d\}.
	\end{aligned}
	\]
	In this instance, $\{a,b,d\}$ does not satisfy JUQ, as we can swap $a$ with $c$. The rules CC, SEQCC, LEXCC, MONROE, the CONSENSUS-RULE, PHRAGMEN-ENESTROEM, and LEXIMINMAXAV can select $\{a,b,d\}$ (among other committees).
	\item[Example 3] 
	$m=4$, $n=4$, $k=2$.
	\[
	\begin{aligned}
		A_1&=\{a,b,c\},\\
		A_2&=\{a,d\},\\
		A_3&=\{a,c,d\},\\
		A_4&=\{a,d\}.
	\end{aligned}
	\]
	In this intance $\{a,c\}$ does not satisfy JUQ, as we can swap $c$ with $d$. The LEXIMAXPHRAGMEN and RSD rules can select $\{a,c\}$.
	\item[\textbf{Example 4}] 
	$m=5$, $n=4$, $k=2$.
	\[
	\begin{aligned}
		A_1&=\{e\},\\
		A_2&=\{a,c,d\},\\
		A_3&=\{c\},\\
		A_4&=\{a,b,c\}.
	\end{aligned}
	\]
	In this instance $\{a,c\}$ does not satisfy JUQ, as we can swap $a$ with $e$. GREEDYMONROE can select $\{a,c\}$ in this instance.
	\item[\textbf{Example 5}] 
	$m=5$, $n=6$, $k=3$.
	\[
	\begin{aligned}
		A_1&=\{a\},\\
		A_2&=\{e\},\\
		A_3&=\{c\},\\
		A_4&=\{b,e\},\\
		A_5&=\{b,c,d\},\\
		A_6&=\{d\}.
	\end{aligned}
	\]
	In this instance $\{b,c,e\}$ does not satisfy JUQ, as we can swap $b$ with $a$. The sequential formulation of Adams (defined analogously as sequential simple Thiele rules, see the definition by \citet[Section 2.3]{LaSk22a}) can select $\{b,c,e\}$ in this instance.
	\item[\textbf{Example 6}] 
	$m=9$, $n=15$, $k=8$.
	\[
	\begin{aligned}
		A_1&=\{a,b,c\},\\
		A_2&=\{a,d,e\},\\
		A_3&=\{a,f,g\},\\
		A_4&=\{b,h\},\\
		A_5&=\{b,h\},\\
		A_6&=\{c,i\},\\
		A_7&=\{c,i\},\\
		A_8&=\{d\},\\
		A_9&=\{d\},\\
		A_{10}&=\{e\},\\
		A_{11}&=\{e\},\\
		A_{12}&=\{f\},\\
		A_{13}&=\{f\},\\
		A_{14}&=\{g\},\\
		A_{15}&=\{g\}.
	\end{aligned}
	\]
	In this instance $\{a,b,c,d,e,f,g\}$ does not satisfy JUQ, as we can swap $a$ with $h$. GJCR (\Cref{def:GJCR}) can select $\{a,b,c,d,e,f,g\}$ in this instance. Note that the violation persists even if we consider the exhaustive committee $\{a,b,c,d,e,f,g,h\}$.

	\item[\textbf{Example 7}] 
	$m=3$, $n=4$, $k=2$.
	\[
	\begin{aligned}
		A_1&=\{a\},\\
		A_2&=\{a,b\},\\
		A_3&=\{a,b,c\},\\
		A_4&=\{a,b,c\}.
	\end{aligned}
	\]
	In this instance $\{a,c\}$ does not satisfy JUQ, as we can swap $c$ with $b$. The Greedy Cohesive Rule introduced by \citet{PPS21a} (in the formulation given by~\citet{BFL+23a}) can select $\{a,c\}$ in this instance.
\end{description}

\subsection{JNQ}\label{app:counterexamples:JNQ}

\begin{description}
	
	\item[\textbf{Example 1}] 
	$m=7$, $n=5$, $k=3$.
	\[
	\begin{aligned}
		A_1&=\{b,e,f\},\\
		A_2&=\{f\},\\
		A_3&=\{a\},\\
		A_4&=\{a,b,c,g\},\\
		A_5&=\{g\}.
	\end{aligned}
	\]
	In this instance $\{a,b,f\}$ does not satisfy JNQ. We can swap $b$ with $g$. The quota of $\voters(b)$ is $\nicefrac{6}{5}$, while the quota of voter $5$ is $\nicefrac{3}{5}$. 
	The rules AV, SAV, CC,  LEXCC, GEOM2, SEQPAV, REVSEQPAV, SEQCC, MINIMAXPHRAGMEN, MAXIMIN-SUPPORT, MONROE, GREEDY-MONROE, MINIMAXAV, LEXIMINMAXAV, EQUAL-SHARES (in all its variants), PHRAGMEN-ENESTROM, can select $\{a,b,f\}$ in this instance.
	
	\item[\textbf{Example 2}] 
	$m=5$, $n=5$, $k=3$.
	\[
	\begin{aligned}
		A_1&=\{b,d\},\\
		A_2&=\{a,d\},\\
		A_3&=\{a,b\},\\
		A_4&= \{c\},\\
		A_5&=\{e\}.
	\end{aligned}
	\]
	In this instance, $\{a,b,d\}$ does not satisfy JNQ, as we can swap $b$ with $e$. The rules PAV, SEQSLAV, SEQPHRAGMEN, RSD and the CONSENSUS-RULE	can select $\{a,b,d\}$.
	
	\item[\textbf{Example 3}] 
	$m=7$, $n=7$, $k=3$.
	\[
	\begin{aligned}
		A_1&=\{a,c\},\\
		A_2&=\{a,b,e\},\\
		A_3&=\{b,c,e\},\\
		A_4&=\{c,e,f,g\},\\
		A_5&=\{d\},\\
		A_6&=\{b,c,e,f,g\},\\
		A_7&=\{a\}.
	\end{aligned}
\]In this instance, $\{a,b,c\}$ does not satisfy JNQ, as we can swap $c$ with $e$. Adams-AV (\Cref{def:adams-av}) and EPH can select $\{a,b,c\}$.	
	\item[\textbf{Example 4}] 
	$m=5$, $n=5$, $k=2$.
	\[
	\begin{aligned}
		A_1&=\{a,b,d\},\\
		A_2&=\{c,d\},\\
		A_3&=\{b,d\},\\
		A_4&= \{e\},\\
		A_5&=\{a,c,d\}.
	\end{aligned}
	\]
	In this instance, $\{b,c\}$ does not satisfy JNQ, as we can swap $b$ with $d$. LEXIMAXPHRAGMEN can select $\{b,c\}$.

\end{description}

\end{document}